%% file: main_arXiv.tex
\documentclass[twocolumn]{article}
\usepackage{arxiv}

\usepackage[utf8]{inputenc} 
\usepackage[T1]{fontenc}    
\usepackage[cmex10]{amsmath}
\usepackage{amssymb}
\usepackage{amsthm}
\usepackage{amsfonts}
\usepackage{bm}
\usepackage{hyperref}       
\usepackage{url}            
\usepackage{booktabs}       
\usepackage{float}
\usepackage{subcaption}
\usepackage{amsfonts}       
\usepackage{nicefrac}       
\usepackage{microtype}      
\usepackage{lipsum}         
\usepackage{graphicx}
\usepackage[numbers]{natbib}
\usepackage{doi}
\usepackage[textsize=small]{todonotes}
\usepackage{comment}
\usepackage{IEEEtrantools}
\usepackage{enumitem}
\usepackage[numbers]{natbib}

\setlist[itemize]{
  leftmargin=1.2em,
  labelsep=0.4em,
  itemsep=2pt,
  topsep=4pt
}

\usepackage{xcolor}

\definecolor{llm}{RGB}{235,242,255}
\definecolor{REPLCOLOR}{RGB}{236,255,240}
\definecolor{statecolor}{RGB}{60,60,60}
\definecolor{codecolor}{RGB}{120,60,160}
\definecolor{subrlmcolor}{RGB}{0,140,140}

\usepackage{tabularx,array}
\newcolumntype{Y}{>{\raggedright\arraybackslash}X} 
\usepackage[ruled,vlined,linesnumbered]{algorithm2e}
\usepackage{listings} 
\SetAlgoVlined
\DontPrintSemicolon

\input{math_definitions.def}

\title{Coded Task Offloading for Fluid Computing: A Privacy-Aware Approach under D2D Networks}

\date{}

\newif\ifuniqueAffiliation
\uniqueAffiliationtrue

\ifuniqueAffiliation 
\author{ 
    \href{https://orcid.org/0009-0002-1110-2099}{\includegraphics[scale=0.06]{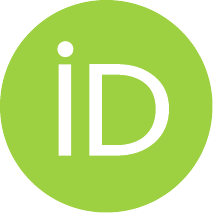}\hspace{1mm}Diego Cajaraville-Aboy}, 
    \href{https://orcid.org/0000-0002-5088-0881}{\includegraphics[scale=0.06]{orcid.pdf}\hspace{1mm}Manuel Fernández-Veiga},
    \href{https://orcid.org/0000-0003-1047-2143}{\includegraphics[scale=0.06]{orcid.pdf}\hspace{1mm}Ana Fernández-Vilas},
    \href{https://orcid.org/0000-0002-2367-2219}{\includegraphics[scale=0.06]{orcid.pdf}\hspace{1mm}Rebeca P. Díaz-Redondo}\\
	atlanTTic -- ICLAB, Universidade de Vigo\\
	Escuela de Ingeniería de Telecomunicación, Vigo, 36310, Spain \\
	\texttt{\{dcajaraville,mveiga,avilas,rebeca\}@det.uvigo.es}
}
\else
\usepackage{authblk}

\newbox{\orcid}\sbox{\orcid}{\includegraphics[scale=0.06]{orcid.pdf}} 

\author[1]{%
	\href{https://orcid.org/0009-0002-1110-2099}{\usebox{\orcid}\hspace{1mm}Diego Cajaraville-Aboy\thanks{\texttt{dcajaraville@det.uvigo.es}}}%
}
\author[1]{%
	\href{https://orcid.org/0000-0000-0000-0000}{\usebox{\orcid}\hspace{1mm}Nombre Supervisor 1\thanks{\texttt{sup1@det.uvigo.es}}}%
}
\author[1]{%
	\href{https://orcid.org/0000-0000-0000-0000}{\usebox{\orcid}\hspace{1mm}Nombre Supervisor 2\thanks{\texttt{sup2@det.uvigo.es}}}%
}

\affil[1]{atlanTTic -- ICLAB, Universidade de Vigo, Escuela de Ingeniería de Telecomunicación, Vigo, 36310, Spain}
\fi

\renewcommand{\headeright}{}
\renewcommand{\undertitle}{}
\renewcommand{\shorttitle}{Coded Task Offloading for Fluid Computing: A Privacy-Aware Approach under D2D Networks}

\hypersetup{
hidelinks,
pdftitle={Coded Task Offloading for Fluid Computing: A Privacy-Aware Approach under D2D Networks},
pdfauthor={Diego Cajaraville-Aboy},
pdfkeywords={},
}

\begin{document}

\twocolumn[{%
\begin{@twocolumnfalse}

\maketitle

\begin{abstract}
Fluid Computing aims to support distributed applications execution across heterogeneous cloud, edge, and device resources, motivating task execution mechanisms that adapt to dynamic and privacy-sensitive environments under runtime conditions. In this context, current task offloading schemes rarely address privacy risks and information leakage under adversarial execution settings; furthermore, most coded computing proposals focus on straggler mitigation without considering system-level objectives such as energy awareness. This paper proposes a coded task offloading scheme for D2D networks under stochastic task arrivals and queue-based dynamics. The proposal combines task offloading techniques with linear secret sharing schemes, where tasks are encoded into redundant shares to support threshold-based recovery, straggler mitigation, and privacy preservation while enhancing system performance. Then, we formulate a privacy-aware offloading problem that jointly optimizes delay and energy while penalizing the theoretical privacy leakage of coded tasks under noisy leakage observations. The problem is solved using a branch-and-bound solver alongside a lightweight heuristic scheduler, both of which are evaluated through a discrete-event simulator. Results show that coded offloading improves the delay--energy trade-off with respect to classical full and parallel offloading schemes, while the heuristic achieves near-optimal performance, outperforming baseline and state-of-the-art solvers. The results also show how privacy leakage penalties reshape offloading decisions, exposing an inherent delay--energy--privacy trade-off.
\end{abstract}


\vspace{1em}

\end{@twocolumnfalse}
}]

\section{Introduction}
\label{sec:introduction}

Current digital systems are increasingly surrounded by distributed and data-driven applications, such as IoT or AI services, running across heterogeneous resources ranging from cloud-native infrastructures and edge servers to resource-constrained devices. These deployments aim to improve response delay, leverage nearby computational capabilities, and adapt execution to dynamic application requirements. Enhancing the deployment and performance of distributed applications is therefore a relevant research challenge, especially in future 6G and D2D scenarios, where communication, computation, and sensing resources are expected to be jointly managed under massive and dynamic environments~\cite{liu2025integrated,wen2024integrated}. Moreover, distributed applications often involve the exchange or processing of sensitive information, which makes privacy and data-management requirements essential, especially in applications such as distributed AI, federated learning, etc.

Several works in the state of the art have addressed this problem within the Cloud-to-Edge Continuum~\cite{gkonis2023survey,ullah2023orchestration}, where computation, storage, and networking resources are distributed across multiple computational tiers to exploit the benefits of each tier and improve latency, energy consumption, or resource utilization. However, the literature has also highlighted that these systems are still limited by fragmented infrastructures, weak interoperability, centralized control assumptions, and poor adaptation to highly dynamic workloads~\cite{marchese2023application,rosmaninho2025edge}. These limitations may lead to suboptimal placements, inefficient communications between tiers and energy consumption, thereby increasing the digital carbon footprint~\cite{patel2024modeling}. In addition, privacy requirements are not always integrated as a first-class design dimension, despite being essential when computation is moved across heterogeneous and potentially untrusted resources.

In recent years, the Fluid Computing paradigm~\cite{lopez2022depth,iorio2022computing} has emerged as a novel approach to overcome these drawbacks. This paradigm abstracts the continuum as a unified compute--storage--network fabric, allowing applications to be deployed according to their requirements while hiding the complexity of the underlying infrastructure. This approach can support fluid deployments in which applications dynamically adapt their placement and execution according to runtime conditions, such as resource availability, mobility, latency, or energy constraints. Even though research in this field is still scarce, especially in terms of algorithmic solutions and orchestration architectures, two promising mechanisms can support and enhance fluid deployments: task offloading and coded computing.

Task offloading~\cite{peng2024survey,pournazari2025computation} has been widely explored as a useful mechanism that allows devices to delegate computational workload to nearby edge servers, cloud nodes, or peer devices in order to improve latency, energy consumption, or execution cost. It has been applied in multiple scenarios, including mobile edge/cloud computing, IoT systems and vehicular networks~\cite{wang2025d2d}. Although task offloading is highly useful, most proposals rely on classical approaches inherited from mobile cloud computing and fog computing settings. As a result, they often do not fully capture the dynamic execution requirements needed for fluid deployments, such as time-varying resource availability and runtime scheduling of communication and computation resources. Moreover, most works focus on system-performance metrics and usually leave privacy concerns as a secondary aspect. Existing privacy-aware task offloading schemes mainly model privacy from the perspective of user metadata or sensitive context information, such as location, task association, or task type, rather than information leakage from the offloaded computational payload~\cite{wang2023location,zhang2024novel}. This limitation becomes more relevant in Fluid Computing environments, where tasks may be dynamically distributed across heterogeneous devices that can be untrusted or only partially observable.

On the other hand, coded computing~\cite{li2020coded,ng2020survey} relies on coding theory to introduce redundancy during workload execution, enabling robustness in distributed computing systems through straggler mitigation while providing data privacy. This idea is conceptually aligned with secure multi-party computation protocols. Linear secret sharing (LSS) schemes, which form a foundational class that cover most of known and practical secret sharing schemes, have also been used as a baseline mechanism to provide redundancy and privacy in secure coded computation~\cite{bitar2020minimizing,schlegel2022privacy}. These approaches have been mainly applied to task-offloading scheme for linear inference and distributed matrix multiplication applications, but they do not often cover general large-scale distributed computation. Moreover, they usually focus on coding mechanisms and recovery guarantees, while they do not fully model general task-offloading dynamics, such as stochastic task arrivals, resource-constrained devices, system evolution or system performance in terms of energy consumption. Privacy concerns are also typically left aside beyond the threshold guarantees of LSS, where the secret can only be recovered after collecting a sufficient number of coded shares. This is especially relevant because even coded data may be exposed through noisy leakage observations, with side-channel attacks being widely recognized as practical cybersecurity threats based on indirect execution signals such as timing, power consumption, or electromagnetic emissions~\cite{beck2024survey,gupta2026security}, along with other potential vectors such as wireless exposure or execution traces.

Even though task offloading and coded computing share common objectives, works that jointly consider both mechanisms, especially under privacy concerns, remain scarce. To address this gap, we propose a privacy-aware coded task offloading scheme for fluid deployments over heterogeneous D2D networks, where tasks are encoded with LSS techniques to enhance system performance through threshold-based recovery from the fastest selected devices. We instantiate the scheme as a time-slotted queue-based scheduling system with stochastic task arrivals, and formulate a delay--energy optimization problem with an explicit theoretical privacy-leakage penalty under noisy leakage observations. Specifically, this paper presents the following contributions:

\begin{itemize}
    \item We model a coded task offloading scheme that combines LSS-based task encoding with time-slotted queue dynamics and stochastic task arrivals over D2D networks.    
    \item We quantify the theoretical information leakage of the LSS-based coded tasks under noisy leakage observations, relying on Fourier-analytic leakage bounds to define the privacy-leakage penalty.
    \item We formulate a privacy-aware optimization problem that jointly minimizes task delay and energy consumption under a privacy-leakage penalty term, and solve it using a branch-and-bound (BnB)-based greedy solver and a lightweight heuristic solver.
    \item We implement a discrete-event simulator to evaluate the proposed system under realistic queue evolution, transmission/computation resource contention, and memory-admission dynamics.
    \item Finally, we analyze the proposed coded task offloading scheme under different validation settings, assessing the impact of the privacy-leakage penalty on the delay--energy trade-off and comparing the proposed solvers with baseline and state-of-the-art task-offloading policies.
\end{itemize}

The remainder of this paper is organized as follows. Section~\ref{sec:related_work} provides a comprehensive state-of-the-art review on task offloading schemes, including privacy-aware approaches. Section~\ref{sec:system_model} outlines the system model for the coded task offloading scheme, including the underlying network, delay, energy consumption, and privacy leakage models. Section~\ref{sec:methodology} defines the considered optimization problem, which is proven to be NP-hard, and the proposed solvers. Section~\ref{sec:experimental_setup} describes the experimental setup and simulation settings, and Section~\ref{sec:results} presents the evaluation results. Finally, Section~\ref{sec:conclusions} concludes the paper, highlighting the contributions and outlining future work.

\section{Related Work}
\label{sec:related_work}

Task/Computation/Workload offloading has been widely studied as a key mechanism to optimize resource allocation in distributed computing networks. These techniques support resource-constrained devices by delegating computation to other nearby devices, edge/fog resources, or cloud capabilities. In general, existing works can be classified according to: i) the execution scenario, including multi-access edge computing, mobile cloud computing, fog-assisted computing, D2D-enabled offloading, etc; ii) the task model, ranging from full task offloading (the entire task is assigned to a single helper node\footnote{A helper node, or simply helper, denotes any computational resource to which a workload is offloaded}) to parallel task execution (task is fragmented into several subtasks to be offloaded to different helper devices); and iii) the considered performance objectives (mainly delay, energy consumption, resource cost, queue stability, or deadline satisfaction). Several surveys~\cite{sadatdiynov2023review} have covered these optimization perspectives and the most common solution techniques for task offloading, including convex optimization, game theory, heuristic methods, and learning-based approaches.

For example, \cite{sufyan2020computation} formulates task offloading in a mobile edge computing scenario as a multi-objective optimization problem to jointly minimize delay, energy consumption, and payment cost. \cite{jia2021energy} studies massive task scheduling in fog--cloud systems by considering the joint impact of energy consumption and execution delay. \cite{niu2025pipelining} considers a pipelining task-offloading strategy in 6G networks to reduce delay, where task scheduling decisions are based on a high-level Cybertwin entity. \cite{khan2023dynamic} formulates task assignment in edge--cloud systems as a graph-based scheduling problem to jointly consider delay, energy, and processing cost. All these works are considered from a high-level orchestration plane that has a complete overview of the system. Other works provide solutions for these (full) task offloading schemes that do not rely on a high-level plane and propose more scalable techniques. \cite{ding2021potential} models the interaction among devices, edge, and cloud resources through a potential-game formulation, where each device selects offloading strategies according to delay and energy costs. \cite{aggarwal2025distributed} proposes a distributed offloading model for MEC systems from a mean-field perspective, focusing on scalable decisions when many devices compete for edge resources. \cite{he2024multi} studies D2D-assisted MEC using multi-agent decision-making to minimize average task delay under deadline constraints. However, most existing works focus on the allocation and scheduling decision itself, abstracting the offloaded workload as directly executable at the selected computing node. This leaves less explored the joint impact of stochastic task arrivals, queue evolution, and runtime resource contention, which are essential to model realistic distributed networks.

Another line of work has considered parallel task offloading, which is especially relevant in IoT and D2D scenarios as computational resources of a single helper may not be enough to execute large workloads efficiently. In this direction, \cite{malik2022efficient} proposes a many-to-one matching-based parallel offloading scheme where subtasks are assigned to fog nodes to reduce task completion delay. Nevertheless, this type of parallel execution requires all subtasks to be completed before recovering the final result. Therefore, even though parallelism can reduce the workload assigned to each helper, the task completion time remains limited by the slowest selected node. On the other hand, coded computing and secret-sharing-based schemes introduce structured redundancy by encoding the workload into coded subtasks or shares, which consequently enables straggler tolerance and threshold-based recovery. For instance, \cite{schlegel2022privacy} proposes a coded MEC scheme for distributed inference, where Shamir's secret sharing protects input data against colluding edge servers while coding mitigates straggling servers. \cite{bitar2021private} studies secure distributed computation based on secret sharing, focusing on the recovery threshold and communication load under privacy and security constraints. Although these coded-computation approaches show how coding can support both straggler mitigation and data privacy, they are usually formulated for specific linear inference or distributed matrix-computation settings, rather than for multi-purpose and energy-aware task offloading schemes.

\begin{table*}[tbh]
\centering
\caption{Comparison of selected task offloading and privacy-aware computation schemes (N/A: not apply)}
\label{tab:related_work_comparison}
\scriptsize
\setlength{\tabcolsep}{2pt}
\renewcommand{\arraystretch}{1.08}
\begin{tabularx}{\textwidth}{@{}l Y Y Y Y Y Y@{}}
\toprule
\textbf{Reference} &
\textbf{Scenario} &
\textbf{Execution model} &
\textbf{Privacy concern} &
\textbf{Privacy metric} &
\textbf{Privacy role} &
\textbf{Optimized criteria} \\
\midrule
\cite{sufyan2020computation} (2020) & MEC/Cloud & Full offloading & Not considered & N/A & N/A & Energy, delay, payment cost \\
\cite{malik2022efficient} (2022) & IoT/Fog & Parallel offloading & Not considered & N/A & N/A & Delay \\
\cite{schlegel2022privacy} (2022) & MEC & Coded offloading & Input data privacy against colluding servers & N/A & Coding guarantee & Delay \\
\cite{wang2023location} (2023) & MEC & Full offloading & Device location privacy during server selection & KL divergence & Objective term & Energy, delay, location privacy \\
\cite{zhang2024novel} (2024) & Edge & Full offloading & Privacy of computational task type & Privacy entropy & Objective term & Energy, delay, task-type privacy \\
\cite{zhang2025joint} (2025) & MEC & Full offloading & Uncertainty in offloading decisions & Privacy entropy & Objective term & Energy, delay, privacy \\
\cite{xia2024privacy} (2024) & MEC & Full offloading & Local data privacy & N/A & Data protection with local differential privacy & Energy, delay \\
\cite{baek2020privacy} (2020) & D2D & Full offloading & Helper's trust & N/A & Redundant offloading and verification mechanisms & Delay \\
\cite{li2021security} (2021) & D2D & Full offloading & Security-critical workload & Security time & Hard constraint & Energy, queue stability \\
\textbf{Our work} & \textbf{D2D} & \textbf{Coded offloading} & \textbf{Information leakage from coded payloads} & \textbf{Mutual information under LSS} & \textbf{Objective penalty} & \textbf{Delay, energy, privacy leakage} \\
\bottomrule
\end{tabularx}
\end{table*}

Privacy-aware task offloading has also received increasing attention. For example, \cite{wang2023location} models privacy-aware task offloading by incorporating a privacy leakage metric, measured with KL divergence, based on the leakage of user position when selecting MEC servers. \cite{yu2023privacy} considers a proxy-assisted MEC offloading to protect both location privacy and task-association privacy during server selection. \cite{zhu2021privacy} studies online task offloading under location privacy leakage, where a specific metric is designed to capture information revealed by sequential decisions. Other proposals quantify privacy through information-theoretic or entropy-based metrics. \cite{zhang2024novel} introduces a privacy-awareness task offloading approach where the uncertainty of the computational task type is quantified through privacy entropy. \cite{zhang2025joint} considers privacy-aware task offloading in healthcare scenarios, where privacy entropy is considered to measure uncertainty of offloading decisions. Differential privacy has also been incorporated into task offloading by perturbing sensitive data or location-related information before making offloading decisions~\cite{xia2024privacy}. In D2D scenarios, \cite{li2021security} considers security-aware collaborative task offloading by introducing a security workload constraint (in terms of CPU capabilities and data size). \cite{baek2020privacy} studies trustworthy and replicated offloading where helper selection reduces repeated offloading to the same helper, while incorporating blockchain-based verification. \cite{razaq2021privacy} studies privacy-aware collaborative task offloading in fog computing, where tasks are allocated according to security credits and requirements. These works show that privacy and security are becoming increasingly relevant in the task offloading field. However, they mainly protect contextual information from the perspective of user metadata, such as location, task association, task type, or helper trust. This leaves mostly unexplored information-theoretic approaches, particularly regarding the information that helper devices may infer from the coded payload they process under noisy leakage observations (such as side-channel attacks). 

Motivated by the lack of task offloading models that jointly integrate coding and information-leakage privacy concerns, this paper models a task offloading from the point of view of a coded computing problem under noisy leakage observations. Specifically, we propose a coded task offloading scheme based on LSS, where each task is encoded into multiple shares and can be recovered after a threshold number of helper computations are completed. Unlike previous offloading works, privacy is modeled as the information leaked by the shares processed by potentially adversarial helper devices. This allows us to quantify the impact of privacy leakage on the delay--energy trade-off of a distributed computing system. Table~\ref{tab:related_work_comparison} summarizes the main differences between representative works and our proposal.

\section{System Model}
\label{sec:system_model}

This section presents the system model of the proposed coded task offloading scheme over D2D networks. We first introduce the D2D network setting and the operational parameters of the time-slotted scheduling system (Subsection~\ref{subsec:network_model}). We then describe the capabilities of the devices and their queue-based evolution dynamics (Subsection~\ref{subsec:device_model}), followed by the communication model for D2D transmissions (Subsection~\ref{subsec:communication_model}). Next, we define the task model, including the LSS mechanism used for coded execution (Subsection~\ref{subsec:task_model}). Finally, we formulate the main performance metrics considered in this work, namely task delay, energy consumption, and privacy leakage under noisy leakage observations (Subsections~\ref{subsec:delay_model}, \ref{subsec:energy_model}, and \ref{subsec:privacy_leakage_model}). The main notations used in this section are summarized in Table~\ref{tab:notations}.

\begin{table}[tbh]
    \centering
    \caption{Summary of Main Notations and Descriptions}
    \label{tab:notations}
    \small
    \setlength{\tabcolsep}{3pt}
    \renewcommand{\arraystretch}{1.02}
    \begin{tabularx}{\columnwidth}{@{} >{\centering\arraybackslash}p{0.26\columnwidth} >{\raggedright\arraybackslash}X @{}}
        \toprule
        \textbf{Notations} & \textbf{Description} \\  \midrule \midrule
        \multicolumn{2}{c}{System Model} \\  \midrule
        $\mathcal{N}$ & Set of $N$ heterogeneous devices   \\  \addlinespace 
        $\mathcal{T}$ & Set of $T$ discrete time slots for task scheduling \\  \addlinespace
        $\tau$ & Fixed duration of each time slot \\  \addlinespace
        $f_i$ & CPU frequency of device $i$ in cycles/second  \\  \addlinespace
        $\kappa_i$ & Effective switched capacitance coefficient of CPU chip architecture of device $i$ \\  \addlinespace
        $M_i$ & Memory capacity of device $i$ for task allocation \\  \addlinespace
        $\mathcal{Q}_i^{\{\text{cmp}/\text{tx}\}}$ & Computation/Transmission queue of device $i$ \\  \addlinespace
        $B_{ij}(t)$ & Bandwidth allocated to the D2D link $i \to j$ \\  \addlinespace
        $R_{ij}(t)$ & Transmission rate for the D2D link $i \to j$ \\  \addlinespace 
        $\mathcal{K}_i(t)$ & Set of arrival tasks from device $i$ at time slot $t$ \\  \addlinespace
        $\ell_{ik}$ & Input data size for task $(i,k)$ in bits \\  \addlinespace
        $c_{ik}$ & CPU cycles per bit for task $(i,k)$ in cycles/bit \\  \addlinespace
        $[n_{ik},s_{ik},t_{ik},q_{ik}]$ & Coding parameters for encoding task $(i,k)$ \\  \addlinespace
        $\psi_{ik}^j$ & Binary decision variable: $1$ if task $(i,k)$ from time slot $t$ is assigned to device $j$ \\  \addlinespace 
        $\mathcal{J}_{ik}$ & Set of helper devices for task $(i,k)$\\  \addlinespace \midrule
        \multicolumn{2}{c}{Performance Metrics} \\  \midrule
        $W_{ik}^{\{\text{cmp}/\text{tx}\}}$ & Waiting time in computation queue and transmission queue \\  \addlinespace
        $T_{ik}^\text{cmp}$ & Computation time for task $(i,k)$ (local or remote subtask) \\  \addlinespace
        $T_{ik}^{\text{tx},(m)}$ & Transmission time for the $m$-th subtask of task $(i,k)$ \\  \addlinespace
        $D_{ik}$ & Total completion delay for task $(i,k)$ \\  \addlinespace
        $E_{ik}$ & Total energy consumption for task $(i,k)$ \\  \addlinespace
        $L_{ik}$ & Privacy leakage risk for task $(i,k)$ based on mutual information \\  \addlinespace
        \bottomrule
    \end{tabularx}
\end{table}

\subsection{Network model}
\label{subsec:network_model}

We consider a D2D network consisting of a set of heterogeneous devices, denoted by $\mathcal{N}=\{1,\dots,N\}$, where each device can communicate directly with any other device through D2D links. The proposed task offloading scheme operates in a time-slotted system, with time slots denoted by $\mathcal{T} = \{1,\dots,T\}$ and coordinated by a high-level scheduler, as illustrated in Figure~\ref{fig:system_model}. For simplicity, each time slot has a fixed duration $\tau$, although the model can be easily extended to variable slot durations. During each time slot $t \in \mathcal{T}$, computational tasks arrive at each device according to a Poisson process with rate $\lambda_i$ for device $i \in \mathcal{N}$, resulting in a per-device set of generated tasks $\mathcal{K}_i(t)$. At the beginning of the next time slot, the scheduler considers all incoming tasks during the previous slot, denoted by $\mathcal{K}(t) \defeq \bigcup_{i\in \mathcal{N}} \mathcal{K}_i(t)$, and they are immediately scheduled (i.e., there is no extra delay between the scheduling epoch and the beginning of the next time slot). The scheduling policy takes into account the incoming tasks and the current system state, which captures the resources already occupied by unfinished tasks. For each incoming task, the scheduler chooses between two decisions: local computation at the source device or offloading to a set of helper devices.

\begin{figure*}[tbh]
    \centering
    \includegraphics[width=1.0\linewidth]{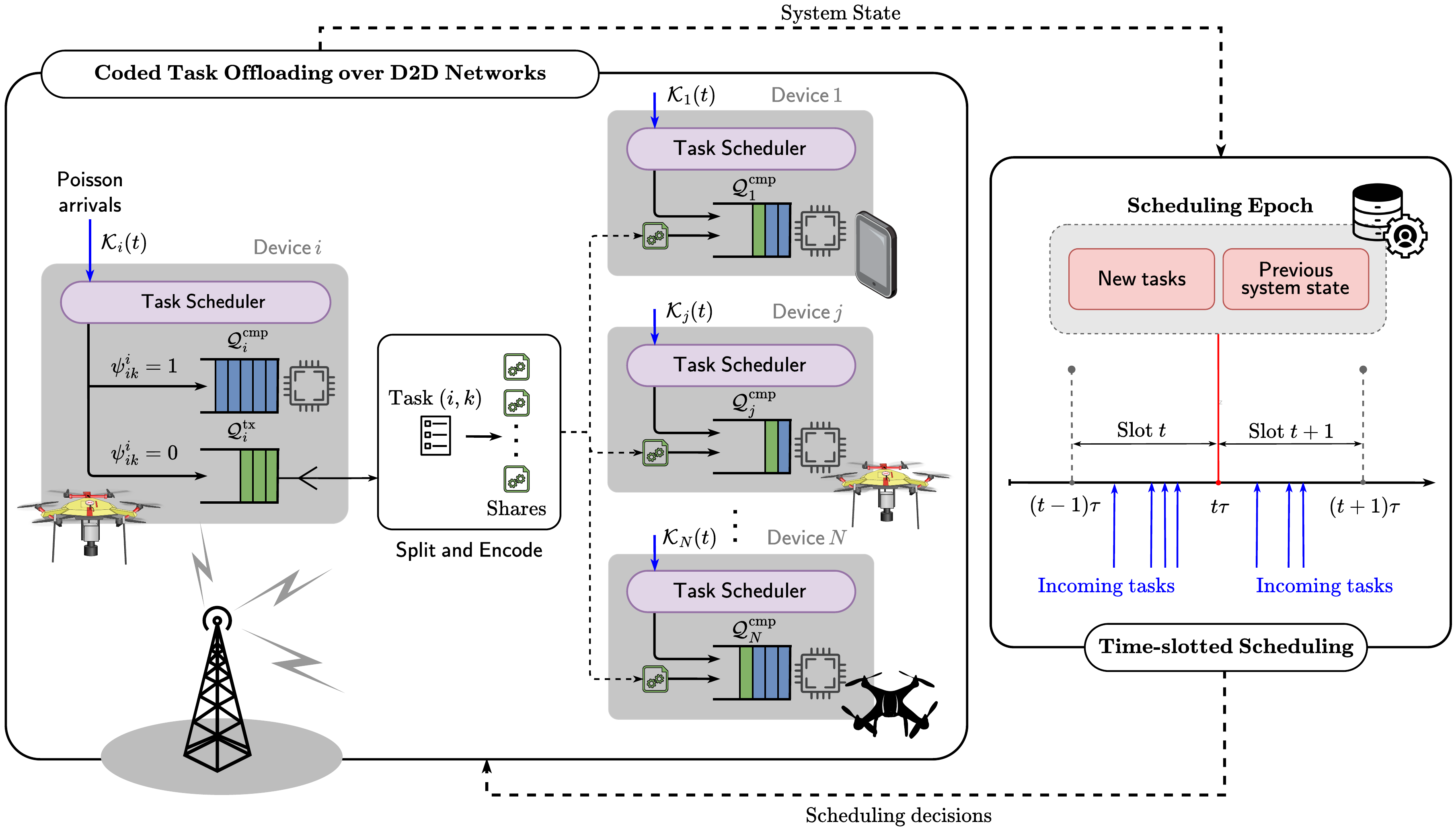}
    \caption{System model of coded task offloading scheme over D2D networks, under a time-slotted scheduling system.}
    \label{fig:system_model}
\end{figure*}

\subsection{Device model}
\label{subsec:device_model}

Each device is modeled as dynamic-queue system, where each device $i$ maintains two FCFS (First-Come, First-Served) single-server queues: computation queue $\mathcal{Q}_i^\text{cmp}$ and transmission queue $\mathcal{Q}_i^\text{tx}$. The device $i$ is defined by the tuple $(f_i, \kappa_i, M_i, \mathcal{Q}_i^\text{cmp}, \mathcal{Q}_i^\text{tx})$, where $f_i$ represents the CPU frequency (in cycles per time unit) for executing tasks, $\kappa_i$ denotes the effective switched capacitance coefficient depending on its CPU chip architecture, $M_i$ is the memory capacity for allocated task data (in bits) in both $\mathcal{Q}_i^\text{cmp}, \mathcal{Q}_i^\text{tx}$. The incoming tasks are allocated in each queue depending on the scheduling decisions by the scheduler. If local execution decision is chosen, the task is queued in $\mathcal{Q}_i^\text{cmp}$ until the device CPU is idle. If remote decision is chosen, the task is queued in $\mathcal{Q}_i^\text{tx}$ until the device transmission unit is idle.

\subsection{Communication model}
\label{subsec:communication_model}

We focus exclusively on direct D2D communications among end devices. The communication layer is modeled as a clustered New Radio (NR) sidelink (PC5) system (within 5G technology), in which a central coordinator assigns orthogonal radio resources to the active D2D transmissions. NR sidelink resources are organized on an OFDM time--frequency grid and scheduled in units of physical resource blocks (PRBs). Accordingly, the coordinator selects the set of active transmitter--receiver pairs and grants each active link a disjoint set of PRBs for its transmission interval, so that simultaneous links do not collide in time--frequency. Each device is equipped with a single transmission interface for task offloading, therefore it cannot handle multiple transmissions in parallel to different helper devices.

In our model, whenever device $i$ is allowed to transmit to device $j$ during time slot $t$, the coordinator assigns a bandwidth portion $B_{ij}(t)$ to link $i \to j$, i.e., the effective set of sidelink radio resources assigned to that transmission. The transmission rate on each directed D2D link $i \to j$ is modeled using the Shannon--Hartley expression for a band-limited additive white Gaussian noise (AWGN) channel
\begin{align}
    R_{ij}(t) = B_{ij}(t) \log_2\left( 1 + \frac{ p_{ij}(t) g_{ij}(t)}{N_0 B_{ij}(t) }\right),
\end{align}
where $p_{ij}(t)$ is the power transmission used by device $i$ to send data to $j$, $g_{ij}(t)$ is the channel gain between devices $i$ and $j$, and $N_0$ is the noise spectral density. The channel gain is modeled through the Friis free-space propagation formula, namely
\begin{align}
    g_{ij}(t)= \left(\dfrac{c}{4\pi f_\mathrm{c} d_{ij}(t)}\right)^2, \quad i\neq j,
\end{align}
where $c$ is the speed of light, $f_\mathrm{c}$ is the carrier frequency, and $d_{ij}(t)$ is the distance between devices $i$ and $j$.

\subsection{Task model}
\label{subsec:task_model}

Each task $(i,k)\in \mathcal{K}_i(t)$ is defined as the tuple $(\ell_{ik}, c_{ik}, [n_{ik},s_{ik},t_{ik},q_{ik}])$, where $\ell_{ik}$ is the input size (in bits) required to perform the task, $c_{ik}$ is the average number of CPU cycles required per input bit, and the parameters $[n_{ik},s_{ik},t_{ik},q_{ik}]$ define the LSS configuration used if task $(i,k)$ is offloaded: $n_{ik}$ is the number of encoded shares\footnote{Hereafter, secret shares are also referred to as coded shares, or simply as shares.} generated from the subtasks, $s_{ik}$ is the number of source subtasks in which the input size is divided, $t_{ik}$ is the recovery threshold, and $q_{ik}$ is the size of the finite field over which the encoding is defined. When task $(i,k)$ is offloaded, its input is first partitioned into $s_{ik}$ source subtasks, which are then encoded into $n_{ik}$ shares using the LSS scheme. Therefore, each selected helper device, distinct from the others, receives one encoded share and executes the corresponding remote computation over masked data, so that no individual helper directly observes the original task input. The task output is recoverable once any $t_{ik}$ of these remote executions finish successfully.

Additionally, we consider a non-perfect secret sharing scheme which may achieve high information rate\footnote{Since the scheme is non-perfect, intermediate sets of shares may reveal partial information, but any set below a privacy threshold reveals no information about the secrets in the ideal threshold-based scheme.}, i.e., the size of each share can be smaller than the size of the secrets~\cite{chen2007secure}. In our model, we adopt the conservative case in which each encoded share has the same size as one source subtask, and therefore yields a worst-case communication and computation cost for the encoded offloading process.
 
At the beginning of time slot $t$, for each task $(i,k)$ in the batch $\mathcal{K}_{i}(t)$, the scheduler chooses a set of binary decision variables denoted by $\psi_{ik}^j \in \{0,1\}, \, j \in \mathcal{N} $. If task $(i,k)$ is executed locally, then $\psi_{ik}^i =1$; otherwise, if it is offloaded, then $\psi_{ik}^i = 0$, and $\psi_{ik}^j = 1$ for $j \neq i$ means that helper device $j$ receives one share of task $(i,k)$. Therefore, the number of selected helpers must match the number of encoded shares: 
\begin{align}
    \sum_{j \neq i} \psi_{ik}^{j} = (1-\psi_{ik}^{i}) n_{ik}, \quad \forall (i,k) \in \mathcal{K}_i(t), \, i \in \mc{N}
\end{align}

\subsection{Task delay model}
\label{subsec:delay_model}

We define the delay of a task as the elapsed time from the beginning of time slot $t$ (when the scheduling decision is made) until the task completes its execution. For a locally executed task, completion occurs when its local computation finishes. For an offloaded task, completion occurs when a sufficient (the recovery threshold) number of encoded shares have been transmitted, queued at their corresponding helper devices, and remotely computed. Depending on the scheduling decision of task $(i,k)$, its task delay is computed as
\begin{align}
    D_{ik} = \psi_{ik}^{i} \, D_{ik}^{\text{loc}} + (1-\psi_{ik}^{i}) \, D_{ik}^{\text{off}},
\end{align}
where $D_{ik}^{\text{loc}}$ is the local task delay, and $D_{ik}^{\text{off}}$ is the offloaded task delay.

\subsubsection{Local execution}
In case of local task execution ($\psi_{ik}^i=1$), no transmission is required and task $(i,k)$ is appended to the computation queue $\mathcal{Q}_i^{\text{cmp}}$. Therefore, the task delay is computed as
\begin{align}
    D_{ik}^{\text{loc}} = W_{ik}^\text{cmp} + T_{ik}^\text{cmp} = W_{ik}^\text{cmp} + \dfrac{\ell_{ik} c_{ik}}{f_i},
\end{align}
where $W_{ik}^\text{cmp}$ is the waiting time in $\mathcal{Q}_i^{\text{cmp}}$ under an FCFS discipline, and  $T_{ik}^\text{cmp}$ is the local computation time.

In terms of queue dynamics, for each device $i$, every job entering $\mathcal{Q}_i^{\text{cmp}}$ (either a local task or a remotely offloaded share) is indexed in order of arrival as $p=1,2,\dots$, with arrival time $\mathrm{a}_p^\text{cmp}$. The corresponding start-of-service time is
\begin{align}
    \mathrm{s}_p^\text{cmp} = 
    \begin{cases}
        \mathrm{a}_1^\text{cmp}, & \text{if} \, p=1\\
        \max\{ \mathrm{a}_p^\text{cmp}, \mathrm{c}_{p-1}^\text{cmp} \}, & \text{if} \, p\geq 2
    \end{cases},
\end{align}
and the completion time is
\begin{equation}
    \mathrm{c}_p^\text{cmp} = \mathrm{s}_p^\text{cmp} + T_p^\text{cmp}.
\end{equation}
Hence, the computation waiting time is
\begin{equation}
    W_{ik}^\text{cmp} = \mathrm{s}_p^\text{cmp} - \mathrm{a}_p^\text{cmp}
\end{equation}
if job $p$ corresponds to a local task $(i,k)$ (or, equivalently, $W_{ik}^{\text{cmp}, (m)}$ if job $p$ corresponds to the $m$-th remote share of task $(i,k)$).

\subsubsection{Task offloading}

On the other hand, in case of remote execution ($\psi_{ik}^i=0$), the $n_{ik}$ encoded shares from task $(i,k)$ through the corresponding LSS scheme, are appended to the transmission queue $\mathcal{Q}_i^{\text{tx}}$. The overall offloading process is summarized in Figure~\ref{fig:coded_task_offloading}. Let
\begin{equation}
    \mathcal{J}_{ik} = \{ j \neq i \, : \, \psi_{ik}^j=1 \} = \{ j_{ik}^{(1)}, \dots, j_{ik}^{(n_{ik})} \},
\end{equation}
denote the set of helper devices selected for task $(i,k)$, so that share $m$ is transmitted from device $i$ to device $j_{ik}^{(m)}$. Since each encoded share carries one $s_{ik}$-fraction of the original task workload, its transmission time of $m$-th share is
\begin{equation}
    T_{ik}^{\text{tx},(m)} = \dfrac{\ell_{ik}}{s_{ik}  R_{i\,j_{ik}^{(m)}}(t)}.
\end{equation}

\begin{figure}[tbh]
    \centering
    \includegraphics[width=0.95\linewidth]{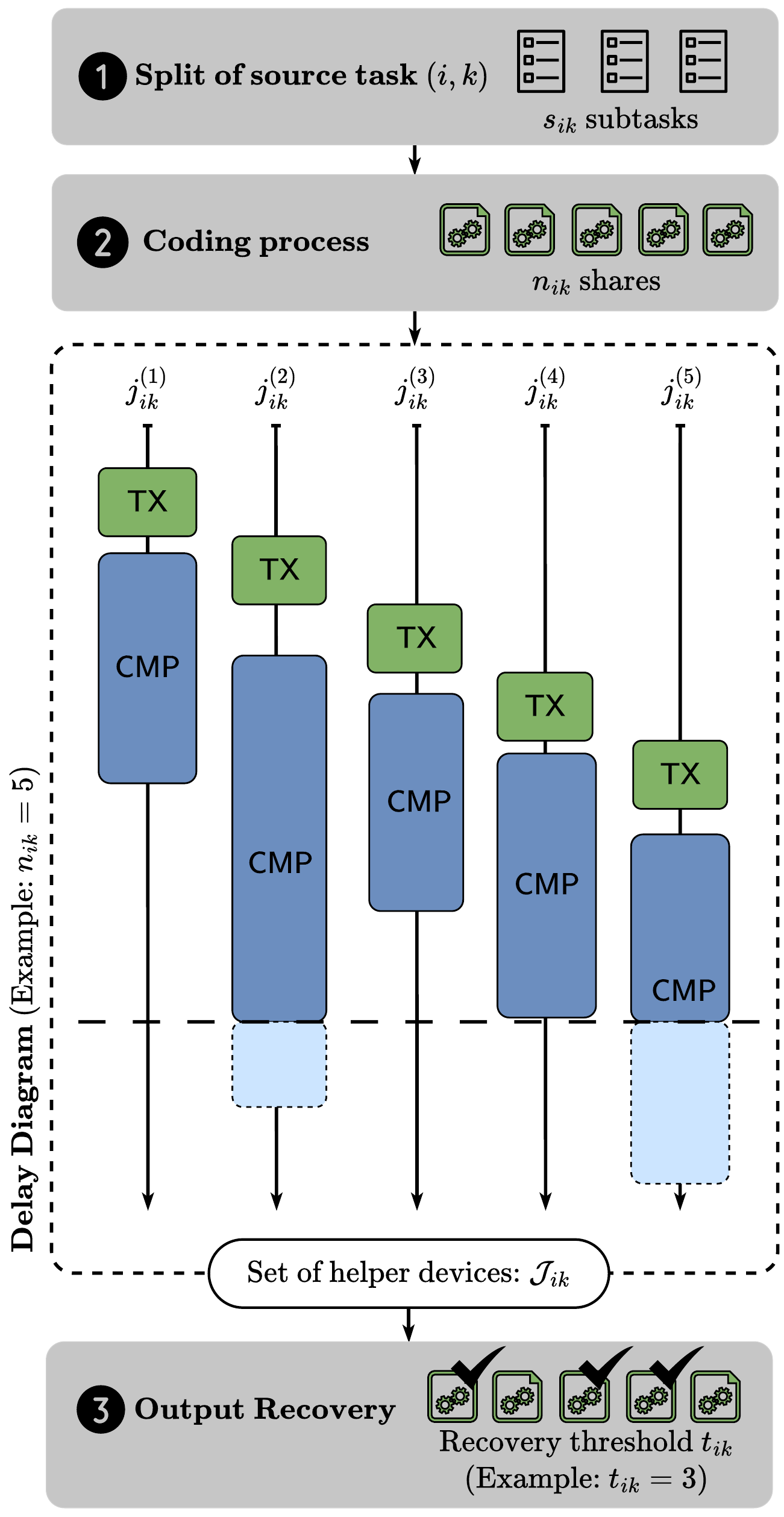}
    \caption{Representation of the coded task offloading scheme}
    \label{fig:coded_task_offloading}
\end{figure}

In terms of queue dynamics, jobs in $\mathcal{Q}_i^{\text{tx}}$ are indexed in order of arrival as $q=1,2,\dots$, with arrival time $\mathrm{a}_q^\text{tx}$, under a FCFS discipline. The start-of-service time is 
\begin{align}
    \mathrm{s}_q^\text{tx} = 
    \begin{cases}
        \mathrm{a}_1^\text{tx} & \text{if} \, q=1\\
        \max\{ \mathrm{a}_q^\text{tx}, \mathrm{c}_{q-1}^\text{tx} \} & \text{if} \, q\geq 2
    \end{cases},
\end{align}
and the completion time of the transmission service is
\begin{equation}
    \mathrm{c}_q^\text{tx} = \mathrm{s}_q^\text{tx} + S_q^\text{tx},
\end{equation}
where $S_q^\text{tx}$ denotes the transmission service time of job $q$. In particular, if job $q$ corresponds to task $(i,k)$, then
\begin{equation}
    S_q^\text{tx} = S_{ik}^{\mathrm{tx}}=\sum_{m=1}^{n_{ik}} T_{ik}^{\mathrm{tx},(m)}.
\end{equation}
Therefore, the waiting time in the transmission queue is 
\begin{equation}
    W_{ik}^\text{tx} = \mathrm{s}_q^\text{tx} - \mathrm{a}_q^\text{tx}.
\end{equation}

Then, the encoded shares are transmitted sequentially to each helper devices but executed remotely in parallel, i.e., once the $m$-th share reaches helper $j_{ik}^{(m)}$, it can join the helper's computation queue and be executed independently of the remaining share transmissions. Since the adopted LSS scheme only requires $t_{ik}$ completed share-results, the task completes when the fastest feasible subset of $t_{ik}$ shares has finished its remote computation. Therefore, the completion delay of an offloaded task $(i,k)$ is
\begin{equation} 
\begin{aligned} 
    D_{ik}^{\mathrm{off}} 
    &= W_{ik}^{\mathrm{tx}} + \min_{\substack{\mathcal{U}\subseteq\{1,\dots,n_{ik}\}\\ |\mathcal{U}|=t_{ik}}} \max_{m\in\mathcal{U}} \\ 
    &\quad \left( \sum_{\ell=1}^{m} T_{ik}^{\mathrm{tx},(\ell)} + W_{ik}^{\mathrm{cmp},(m)} + T_{ik}^{\mathrm{cmp},(m)} \right). 
\end{aligned}
\end{equation}
where $W_{ik}^{\mathrm{cmp},(m)}$ and 
\begin{align}
    T_{ik}^{\mathrm{cmp},(m)} = \frac{\ell_{ik} \, c_{ik}}{s_{ik} \, f_{j_{ik}^{(m)}}}
\end{align}
are, respectively, the waiting time and computation time of $m$-th share $m$ at helper $j_{ik}^{(m)}$. 

\subsection{Task energy consumption model}
\label{subsec:energy_model}

We define the energy consumption of a task as the total amount of energy consumed by the devices involved in its execution, including local/remote computation and, when applicable, wireless transmission.

In task-offloading models, the energy consumed by a CPU to process a workload is commonly represented through dynamic power models. Accordingly, if task $(i,k)$ is executed locally ($\psi_{ik}^i=1$), the required number of CPU cycles is $\ell_{ik}\, c_{ik}$ and, therefore, the local computation energy is given by
\begin{equation}
    E_{ik}^{\mathrm{loc}}=\kappa_i \, f_i^2 \, \ell_{ik} \, c_{ik}.
\end{equation}

Conversely, when the task is offloaded ($\psi_{ik}^i=0$), the total energy consumption depends on: i) the wireless transmission energy required to send the encoded shares to the selected helper devices; and ii) the required remote computation energy consumed at the helper devices. Under previous considerations, the total transmission energy is
\begin{equation} 
    E_{ik}^{\mathrm{tx}} = \sum_{m=1}^{n_{ik}} p_{i\,j_{ik}^{(m)}}(t)\,T_{ik}^{\mathrm{tx},(m)} \\ 
    = \sum_{m=1}^{n_{ik}} \frac{\ell_{ik}\,p_{i\,j_{ik}^{(m)}}(t)}{s_{ik}\,R_{i\,j_{ik}^{(m)}}(t)}. 
\end{equation}

For the computation energy consumption, we adopt a cancellation-aware mechanism, i.e., when the fastest $t_{ik}$ completed share-results are available, the remaining share executions are cancelled immediately. Let
\begin{align}
    \mathrm{b}_{ik}^{(m)} =
    \sum_{\ell=1}^m T_{ik}^{\mathrm{tx},(\ell)} + W_{ik}^{\mathrm{cmp},(m)}
\end{align}
be the instant, measured after the task starts its transmission service, at which $m$-th share starts its remote computation. Also, let
\begin{align}
    \mathrm{d}_{ik}^{\mathrm{off}} =
    D_{ik}^{\mathrm{off}}-W_{ik}^{\mathrm{tx}}
\end{align}
be the offloaded-task completion time measured from the start of the task transmission service. Then, the actual computation time spent by helper $j_{ik}^{(m)}$ on share $m$ before completion/cancellation is
\begin{equation} 
    \widetilde{T}_{ik}^{\mathrm{cmp},(m)} 
    = \Bigl[ \min\bigl\{ \mathrm{b}_{ik}^{(m)}+T_{ik}^{\mathrm{cmp},(m)},\, \mathrm{d}_{ik}^{\mathrm{off}} \bigr\} - \mathrm{b}_{ik}^{(m)} \Bigr]^+. 
\end{equation}
where $[x]^+\triangleq\max\{x,0\}$. This expression gives zero energy if share $m$ has not started execution before the task is reconstructed, full computation energy if it finishes before reconstruction, and partial computation energy if it is cancelled while being executed. Therefore, the remote computation energy consumed by helper devices is
\begin{align}
    E_{ik}^{\mathrm{rem}} =
    \sum_{m=1}^{n_{ik}}
    \kappa_{j_{ik}^{(m)}}
    f_{j_{ik}^{(m)}}^3
    \widetilde{T}_{ik}^{\mathrm{cmp},(m)},
    \label{eq:rem_energy_cancellation}
\end{align}
which accounts only for the fraction of the computation actually performed before task completion.

Hence, the total energy consumption of an offloaded task $(i,k)$ is
\begin{align}
    E_{ik}^{\mathrm{off}} = E_{ik}^{\mathrm{tx}}+E_{ik}^{\mathrm{rem}},
\end{align}
and, therefore, combining both scheduling actions, the energy consumption of task $(i,k)$ is expressed as
\begin{align}
    E_{ik} =\psi_{ik}^i \, E_{ik}^\text{loc} + (1-\psi_{ik}^i) \, E_{ik}^\text{off}.
\end{align}

\subsection{Privacy Leakage Model}
\label{subsec:privacy_leakage_model}

\textbf{Threat Model:} We assume honest-but-curious helper devices that execute the assigned coded computations correctly, but may attempt to infer information about the original subtask inputs from their received shares. A passive adversary may collect such observations and is assumed to know the coding scheme and the system model, but cannot alter computations or inject malicious results.

In this context, we quantify the theoretical information leakage about the original source subtasks after the coded offloading procedure. This leakage value will be used to study the privacy--aware system performance trade-off induced by coded task offloading.

We consider the setting of LSS schemes over a finite field $\mathbb{F}_q$, where $q = p^d$ is an integer power of a prime $p$. Our goal is to quantify the information leakage of the LSS scheme, between the secret vector $\mathbf{S}=(S_1,\dots,S_k) \in \mathbb{F}_q^k$ (where each $S_i$ is an iid uniform random variable in $\mathbb{F}_q$) and the leakage vector $\mathbf{L}=(L_1,\dots,L_n) \in \mathbb{F}_q^n$ observed by the adversary from the shares $\mathbf{U}=(U_1,\dots,U_n) \in \mathbb{F}_q^n$ (which is a LSS for $\mathbf{S}$). Specifically, we instantiate the offloading protection through a ramp-LSS scheme because they trade perfect privacy against lower overhead and higher information rate (non-perfect LSS schemes), which is better aligned with dynamic and resource-constrained D2D environments. Remark that, to avoid overloading the task index notation for task $(i,k)$, the correspondence is $n=n_{ik}$, $k=s_{ik}$, $t=t_{ik}$, and $q=q_{ik}$.

\begin{definition}[$(t,k,n)$ Ramp Secret Sharing scheme]
\label{def:ramp_secret_sharing_scheme}
Let $\mathcal{C}_2 \subsetneq \mathcal{C}_1 \subseteq \mathbb{F}_q^n$ be two nested MDS codes where $\mathcal{C}_{1}$ is an $[n,t,n-t+1]_{q}$ code, and $\mathcal{C}_{2}$ is an $[n,t-k,n-t+k+1]_{q}$, which satisfy $k < t \leq n$. The scheme $\mathrm{RampSS}(t,k,n)$ distributes $n$ shares of the secret vector $\mathbf{S} \in \mathbb{F}_{q}^{k}$ such that: 
\begin{itemize}
    \item[i)] any set of at most $t-k$ shares reveals no information about $\mathbf{S}$;
    \item[ii)] any set of at least $t$ shares reconstructs $\mathbf{S}$.
\end{itemize}
To generate the shares, choose a linear complement code $\mathcal{S}$ of $\mathcal{C}_2$ in $\mathcal{C}_1$, i.e., $\mathcal{C}_1=\mathcal{S} \oplus \mathcal{C}_2$, with $\mathcal{S} \cap \mathcal{C}_2 = \{\mathbf{0}\}$, $\dim(\mathcal{S})=k$; and fix an isomorphism $\psi : \mathbb{F}_q^k \to \mathcal{S}$. For a secret $\mathbf{s} \in \mathbb{F}_{q}^{k}$, the isomorphism $\psi$ selects the coset $\psi(\mathbf{s})+\mathcal{C}_{2} \in \mathcal{C}_{1}/\mathcal{C}_{2}$, and the share vector $\mathbf{U}=(U_{1},\dots,U_{n})$ is sampled uniformly at random from this coset.
\end{definition}

Nested MDS codes give a closed-form threshold structure, since the privacy threshold is $t-k$, the reconstruction threshold is $t$, and the straggler tolerance of the offloaded task is $n-t$.

On the other hand, we consider the total variation distance as the baseline privacy leakage metric to measure the statistical distinguishability between the joint distribution of the secret and the leakage and the product of their marginals.

\begin{definition}[Pointwise Total Variation Distance] 
The pointwise total variation (TV) distance between two random vectors $\mathbf{S}$ and $\mathbf{L}$, over alphabets $\mathcal{S}$ and $\mathcal{L}$, respectively, is defined as
\begin{equation*} 
\begin{split} 
    \Delta^{\mathrm{TV}}(\mathbf{S};\mathbf{L}) &= \frac{1}{2} \sum_{\mathbf{s}\in\mathcal{S}} \sum_{\boldsymbol{\ell}\in\mathcal{L}} \Bigl| \mathbb{P}(\mathbf{S}=\mathbf{s},\mathbf{L}=\boldsymbol{\ell}) \\ 
    &\hspace{3.2em} - \mathbb{P}(\mathbf{S}=\mathbf{s})\,\mathbb{P}(\mathbf{L}=\boldsymbol{\ell}) \Bigr|. 
\end{split} 
\end{equation*}
\end{definition}

We assume that each share is leaked through a noisy leakage observation, which is appropriate for our setting since the shares are transmitted to and processed by different helper devices. 

\begin{definition}[$\delta$-Noisy Channel, \cite{gupta2026security}]
    For some random variable $U$ over $\mathcal{U}$, a (randomized) leakage function $L \, : \, \mathcal{U} \to \mathcal{L}$ is said to be $\delta$-noisy leakage channel if $\Delta^{\text{TV}}(U;L) \leq \delta$. The leakage output will be called a $\delta$-noisy observation of $U$.
\end{definition}

The leakage parameter of each helper, $\delta_i$, captures the statistical strength of the corresponding leakage observation, and can be interpreted as a privacy risk across helper devices. Under usual assumptions, we consider that the channels $(L_i \mid U_i)$ are mutually independent for all $i$, and that these channels are subject to independent observations such that $\Delta^{\text{TV}}(U_i; L_i) \leq \delta_i$. Under the LSS scheme from Definition~\ref{def:ramp_secret_sharing_scheme} and the previously cited leakage model, we obtain the following theorem.

\begin{theorem}
\label{thm:privacy_leakage_bound}
Consider a $\mathrm{RampSS}(t,k,n)$ scheme over $\mathbb{F}_{q}$ under the $\delta$-noisy share leakage model. Let $\mathbf{S} \in \mathbb{F}_{q}^{k}$ be the secret vector, $\mathbf{U}=(U_{1},\dots,U_{n})$ the corresponding share vector, and $\mathbf{L}=(L_{1},\dots,L_{n})$ the leakage vector. If $\Delta^{\mathrm{TV}}(U_{j}; L_{j}) \leq \delta_{j}$ for every share
$j \in [n]$, then
\begin{equation} 
    \Delta^{\mathrm{TV}}(\mathbf{S};\mathbf{L}) 
    \leq 2^{t-k+1} \left(q^{n-t+k}-1\right) \cdot \max_{\mathcal{J} \in \binom{[n]}{t-k+1}} \prod_{j \in \mathcal{J}} \delta_j,
\label{eq:privacy_leakage_bound} 
\end{equation}
where $[n]=\{1,\dots,n\}$ and $\binom{[n]}{x}$ denotes the collection of all $x$-element subsets of $[n]$.
\end{theorem}

The proof follows the Fourier-analysis approach for noisy channel-side leakage in linear
secret sharing schemes introduced in \cite{gupta2026security}, adapted to the
multisecret and nested structure of $\mathrm{RampSS}(t,k,n)$. The full derivation is provided in Appendix~\ref{app:privacy_leakage_derivation}. For each offloaded task $(i,k)$, we consider the conservative value of the privacy leakage upper bound to quantify the privacy cost associated with offloading task $(i,k)$ to the selected helper devices, which is denoted by $L_{ik}$. Remark that, if the task $(i,k)$ is executed locally, we set $L_{ik}=0$.

\section{Methodology}
\label{sec:methodology}

In this section, we formulate the non-linear optimization problem that captures the objective of our coded task offloading scheme (Subsection~\ref{subsec:problem_formulation}). Since the problem is proven to be NP-hard, we propose two algorithmic solutions: a BnB-based greedy solver (Subsection~\ref{subsec:greedy_bnb}) to obtain near-optimal solutions, and a heuristic solver (Subsection~\ref{subsec:heuristic}) to provide a lightweight scheduling policy for evaluating system performance.

\subsection{Problem formulation}
\label{subsec:problem_formulation}

Our objective is to jointly minimize the long-term average task delay and energy consumption while accounting for the privacy leakage induced by offloading secret-shared tasks to helper devices. In our optimization problem, the privacy cost is introduced as a penalty term in the objective function that provides decisions to the scheduler to explicitly trade delay and energy savings against the privacy leakage of each offloading decision.

Since this results in a multi-objective optimization problem, we use a scalarization method to define a unique objective function as a convex combination of both individual objectives (delay and energy consumption). Additionally, to avoid the dependence of the objective on problem-specific maxima, we normalize delay and energy through smooth scale functions:
\begin{equation} 
\begin{aligned} 
    \widehat{D}_{ik} 
    &= \frac{2}{\pi}\arctan\!\left(\frac{D_{ik}}{D_0}\right), \\ 
    \widehat{E}_{ik} &= \frac{2}{\pi}\arctan\!\left(\frac{E_{ik}}{E_0}\right), 
\end{aligned} 
\end{equation}
where $D_0$ and $E_0$ are reference task delay and task energy consumption values, respectively. On the other hand, we define the privacy penalty function as
\begin{equation}
    \Pi\!\left( L_{ik} \right) = \sigma \left( \log_2(L_{ik}) / \xi  \right),
\end{equation}
where $\sigma(\cdot)$ is the sigmoid function, and $\xi>0$ controls the smoothness of the transition. This function provides a smooth and upper-bounded privacy penalty that, when the leakage grows, the penalty increases smoothly.

Let us define the joint decision variable $\bm{\Psi}(t) \defeq (\Psi_i(t))_{i=1}^N$ where $\Psi_i (t) = (\psi_{ik}^j)_{k,j} \in \mbb{M}_{|\mathcal{K}_i(t)| \times N} $ as the set of scheduling decisions at time slot $t$. We define the long-term optimization problem $\mathcal{P}1$ as:
\begin{subequations}
\begingroup
\begin{IEEEeqnarray}{rCl}
\underset{\{\bm{\Psi}(t)\}_{t \in \mathcal{T}}}{\operatorname{minimize}}\;
& &
\limsup_{T\to\infty}
\frac{1}{T}\sum_{t=1}^{T}
\frac{1}{|\mc{K}(t)|}
\sum_{(i,k)\in\mc{K}(t)}
\bigl(
\alpha \hat{D}_{ik}
\nonumber\\
& &
{} \hspace{5em} +\beta \hat{E}_{ik}
+\gamma \Pi(L_{ik})
\bigr)
\label{eq:problem_obj}
\\[0.2ex]
\text{s.t.}\;
& &
|\mc{Q}_i^{\mathrm{tx}}(t)|
+
|\mc{Q}_i^{\mathrm{cmp}}(t)|
\leq M_i,
\,\, i\in\mc{N}
\label{eq:problem_a}
\\[0.5ex]
&
&
\psi_{ik}^{j}\in\{0,1\},
\,\, (i,k)\in\mc{K}(t),\ j\in\mc{N}
\label{eq:problem_b}
\\[0.5ex]
&
&
\sum_{j\in\mc{N}\setminus\{i\}} \psi_{ik}^{j}
=
(1-\psi_{ik}^{i}) n_{ik}.
\label{eq:problem_c}
\end{IEEEeqnarray}
\endgroup
\end{subequations}

Problem $\mathcal{P}1$ is a binary non-linear optimization problem (BNLP) with dynamic queue-state dependence. It captures: (i) a time-slotted scheduling system under queue-based dynamics under stochastic task arrivals; (ii) a coded task offloading scheme based on LSS schemes and coded-computing principles, where each offloaded task can be recovered after the completion of a subset of offloaded shares; and (iii) a privacy-leakage-based penalty into the delay--energy trade-off through a smooth regularization term, so that offloading decisions are penalized according to the privacy leakage induced by the selected LSS scheme under noisy leakage observations. 

The main difficulty stems from the combinatorial nature of the binary scheduling process since the number of joint scheduling decisions grows exponentially with the number of tasks and candidate helpers. Therefore, although optimal solutions can be achieved by exhaustive exploration, solving $\mathcal{P}1$ exactly at each slot is, in general, computationally impractical.

\begin{proposition}
Problem $\mathcal{P}1$ is NP-hard.
\end{proposition}
We refer the reader to Appendix~\ref{app:np_hardness} for the proof.

In the following, we use a task-action cost notation for both proposed solvers. For task $(i,k)$, an action $a$ denotes either local execution or offloading to a helper set. The state $\mathbf{S}$ denotes the queue and resource state at the moment when the action is evaluated. The cost of assigning task $(i,k)$ through action $a$ under state $\mathbf{S}$ is defined as
\begin{align}
    c_{ik}(a\mid\mathbf{S})
    \defeq
    \alpha \hat{D}_{ik}(a\mid\mathbf{S})
    &+ \beta \hat{E}_{ik}(a\mid\mathbf{S}) \nonumber\\
    &+ \gamma \Pi\!\left(L_{ik}(a\mid\mathbf{S})\right),
    \label{eq:task_action_cost}
\end{align}
where $\hat{D}_{ik}(a\mid\mathbf{S})$, $\hat{E}_{ik}(a\mid\mathbf{S})$, and $L_{ik}(a\mid\mathbf{S})$ are the normalized delay, normalized energy, and privacy leakage induced by that action, respectively. If an action violates memory, queue, or helper-availability constraints, its cost is set to $+\infty$.

\subsection{Branch-and-bound-based greedy scheduler}
\label{subsec:greedy_bnb}

We introduce a BnB-based greedy solver\footnote{Branch-and-bound performs an implicit enumeration of the feasible decision tree, pruning partial schedules whose lower bound cannot improve the incumbent solution.} that serves as a strong baseline for validating the system performance against lower-complexity scheduling policies. This approach combines two ideas: i) it constructs a deterministic reduced action set for each task, preserving the main properties of the model; and ii) it applies a depth-first search over the resulting decision space. BnB method provides an exact solution over the reduced action space while avoiding exhaustive enumeration. Algorithm~\ref{alg:reduced_bnb} summarizes the proposed procedure.

\begin{algorithm}[tbh]
\LinesNotNumbered
\caption{Branch-and-bound-based Greedy Solver}
\label{alg:reduced_bnb}
\KwIn{task batch $\mathcal{K}(t)$, initial state $\mathbf{S}_0(t)$}
\KwOut{slot decision matrix $\bm{\Psi}(t)$}

Construct $\widetilde{\mathcal{A}}_{ik}(t)$ for every $(i,k) \in \mathcal{K}(t)$ using \eqref{eq:reduced_action_set}\;

Order the tasks as $\pi(1), \dots, \pi(|\mathcal{K}(t)|)$\;

\For{$\ell=1, \dots, |\mathcal{K}(t)|$}{
    rank the actions in $\widetilde{\mathcal{A}}_{\pi(\ell)}(t)$ by the root-state cost $c_{\pi(\ell)}(a\mid\mathbf{S}_0(t))$\;
    define the ordered action list $\widetilde{\mathcal{A}}^{\uparrow}_{\pi(\ell)}(t)$ accordingly\;
    compute $c^{\mathrm{min}}_{\pi(\ell)}(t)$ using \eqref{eq:root_state_min_cost}\;
}

Compute the sequence $\{B_d(t)\}_{d=1}^{|\mathcal{K}(t)|}$ using \eqref{eq:static_lower_bound}\;

Build the initial feasible schedule using \eqref{eq:greedy_rollout_action}, and set its total cost as the initial incumbent $U(t)$\;

\BlankLine
\textbf{Recursive step at a current node $v$ of depth $d(v)$:}\;

Compute the accumulated cost $G_{v}(t)$ of the current partial schedule\;

\If{$G_{v}(t)+B_{d(v)}(t)\ge U(t)$}{
    prune node $v$\;
}
\Else{
    \ForEach{$a\in \widetilde{\mathcal{A}}^{\uparrow}_{\pi(d(v))}(t)$}{
        evaluate $c_{\pi(d(v))}(a \mid \mathbf{S}_{v}(t))$\;
        \If{$c_{\pi(d(v))}(a \mid \mathbf{S}_{v}(t)) < \infty$}{
            append $a$ to the partial schedule\;
            \If{$d(v)=|\mathcal{K}(t)|$}{
                update the incumbent solution and set $U(t)\gets G_v(t)+c_{\pi(|\mathcal{K}(t)|)}(a \mid \mathbf{S}_{v}(t))$ if $G_v(t)+c_{\pi(|\mathcal{K}(t)|)}(a \mid \mathbf{S}_{v}(t)) < U(t)$\;
            }
            \Else{
                update to successor state $\mathbf{S}_{v'}(t)$\;
                recurse at the child node $v'$ with depth $d(v')=d(v)+1$\;
            }
        }
    }
}
\Return $\bm{\Psi}(t)$\;
\end{algorithm}

\subsubsection{Deterministic reduced action set}

Let $\mathbf{S}_0(t)$ denote the system state at the beginning of time slot $t$. The full feasible action set of task $(i,k)$ is
\begin{align}
    \mathcal{A}^{\mathrm{full}}_{ik}(t) = \{i\} \cup \Bigl\{ \mathcal{J}\subseteq \mathcal{N}\setminus\{i\} : |\mathcal{J}| = n_{ik} \Bigr\},
    \label{eq:full_action_set}
\end{align}
where $\{i\}$ denotes local execution, while each set $\mathcal{J}$ denotes a coded offloading action using $n_{ik}$ helpers. 

A pure BnB search over $\mathcal{A}^{\mathrm{full}}_{ik}(t)$ remains computationally prohibitive, since the number of offloading actions for a task is $\binom{N-1}{n_{ik}}$. To alleviate this issue, we construct a deterministic reduced action set. For every candidate helper $j \neq i$, we compute the candidate cost $c_{ik}(\{j\} \mid \mathbf{S}_0(t))$. Let $\mathcal{H}_{ik}(t)$ be the set containing the $\min\{N-1,\,3n_{ik}\}$ helpers with smallest values of the candidate costs. This keeps the best-ranked helpers while preserving non-trivial combinatorial diversity. Then, the reduced action set is defined as
\begin{equation}
    \widetilde{\mathcal{A}}_{ik}(t) = \{i\} \cup \Bigl\{ \mathcal{J} \subseteq \mathcal{H}_{ik}(t) : |\mathcal{J}| = s_{ik} \Bigr\}.
    \label{eq:reduced_action_set}
\end{equation}

\subsubsection{Branch-and-bound search}

The BnB search is conducted over the reduced action sets by a depth-first search (DFS) procedure. Let $\pi(1), \dots, \pi(|\mathcal{K}(t)|)$ be the task ordering used by the search, where tasks with larger number of shares and larger workload $\ell_{ik}c_{ik}$ are processed first. The search tree is composed of nodes, each of them represents one partial task schedule. Let $d(v) \in \{1, \dots, |\mathcal{K}(t)|\}$ denote the depth of node $v$, meaning that the first $d(v)-1$ ordered tasks have already been assigned. At node $v$, the algorithm branches on one action of task $\pi(d(v))$, evaluated under the current partial state $\mathbf{S}_{v}(t)$.

The pruning mechanism is based on three quantities: i) the accumulated cost $G_v(t)$ of the current partial solution represented by node $v$, computed as the sum of the action costs along the current branch, ii) an incumbent upper bound $U(t)$ associated with the best complete schedule found so far, and iii) a static lower bound $B_{d(v)}(t)$ for the current task and the remaining tasks. Thus, node $v$ is pruned whenever
\begin{align}
    G_v(t) + B_{d(v)}(t) \ge U(t).
    \label{eq:bnb_pruning_rule}
\end{align}

Before the tree expansion starts, we compute, for each ordered task $\pi(\ell)$, the minimum root-state cost
\begin{equation}
    c^{\mathrm{min}}_{\pi(\ell)}(t) \defeq \min_{a \in \widetilde{\mathcal{A}}_{\pi(\ell)}(t)} c_{\pi(\ell)}(a \mid \mathbf{S}_0(t)),
    \label{eq:root_state_min_cost}
\end{equation}
and define the static lower bound at depth $d$ as
\begin{equation}
    B_d(t) \defeq \sum_{\ell=d}^{|\mathcal{K}(t)|} c^{\mathrm{min}}_{\pi(\ell)}(t), \quad d=1,\dots,|\mathcal{K}(t)|.
    \label{eq:static_lower_bound}
\end{equation}
Hence, $B_d(t)$ depends only on the depth and not on the particular node at that depth. We choose this reference bound because it evaluates every remaining task independently from the slot-start state, hence ignoring the additional congestion generated by future assignments, i.e., the lower bound value cannot exceed the completion cost of any descendant node. In fact, evaluating lower bounds in which the remaining tasks are re-evaluated from the current partial state at every node produces the same scheduling decisions on the reduced action space while incurring substantially larger runtime.

To initialize the search, the algorithm first constructs an initial feasible schedule by a greedy search. Starting from $\mathbf{S}_0(t)$, for $\ell = 1,\dots,|\mathcal{K}(t)|$, it selects
\begin{equation}
    a_\ell^{(0)} \in \arg \min_{a \in \widetilde{\mathcal{A}}_{\pi(\ell)}(t)} c_{\pi(\ell)} \left( a \mid \mathbf{S}_{\ell-1}(t) \right),
    \label{eq:greedy_rollout_action}
\end{equation}
updates the state to $\mathbf{S}_{\ell}(t)$, and continues sequentially until all tasks are assigned. The corresponding total cost defines the initial incumbent bound, denoted by $U(t)$. During the DFS-based BnB search, this incumbent $U(t)$ is iteratively updated whenever a complete schedule with smaller total cost is found. At the end of the search, the final scheduling decision for slot $t$ is the incumbent complete assignment $\bm{\Psi}(t)$ associated with the minimum total cost $U(t)$ among all feasible schedules explored on the reduced action space.

\subsubsection{Complexity analysis}

The algorithm is exact on the reduced action space $\widetilde{\mathcal{A}}_{ik}(t)$, hence sub-optimal with respect to the original problem. The preprocessing stage evaluates $N-1$ candidate costs per task and ranks the corresponding helpers, so its cost is polynomial in $N$ plus the generation of the reduced actions $|\widetilde{\mathcal{A}}_{ik}(t)|$. Since the latter dominates over the former, the DFS-based BnB stage has worst-case time complexity $\mathcal{O} \left(  \prod_{(i,k) \in \mathcal{K}(t)} |\widetilde{\mathcal{A}}_{ik}(t)| \right)$ with
\begin{equation*}
    |\widetilde{\mathcal{A}}_{ik}(t)| = 1+\binom{\min\{N-1,3n_{ik}\}}{n_{ik}},
    \label{eq:reduced_action_cardinality}
\end{equation*}
which remains exponential (as expected for combinatorial search). However, this is substantially smaller than the corresponding full-action complexity $ \mathcal{O} \bigl(\prod_{(i,k)} [1 + \binom{N-1}{n_{ik}}] \bigr)$, and the pruning rule \eqref{eq:bnb_pruning_rule} further reduces the explored tree in practice. 

\subsection{Heuristic solver}
\label{subsec:heuristic}

We also propose a heuristic solver that avoids combinatorial search and builds a single schedule sequentially. Specifically, the heuristic solver follows this workflow: i) it prioritizes tasks according to the potential gain of offloading at the initial state; ii) then, for each ordered task in a sequential way, it compares local execution cost against offloading costs with the best helper set selected by candidate costs, to fix the best action. Algorithm~\ref{alg:heuristic_scheduler} summarizes the procedure.

\begin{algorithm}[t]
\LinesNotNumbered
\caption{Heuristic Solver}
\label{alg:heuristic_scheduler}
\KwIn{task batch $\mathcal{K}(t)$, initial state $\mathbf{S}_0(t)$}
\KwOut{slot decision matrix $\bm{\Psi}(t)$}

\ForEach{$(i,k)\in\mathcal{K}(t)$}{
    form $\widehat{\mathcal{J}}^{(0)}_{ik}$ with the $n_{ik}$ helpers having candidate costs $c_{ik}(\{j\}\mid\mathbf{S}_0(t))$\;
    compute the priority score $\Delta_{ik}$ using \eqref{eq:heuristic_priority_score}\;
}

Order the tasks as $\pi(1), \dots, \pi(|\mathcal{K}(t)|)$ in non-increasing order of $\Delta_{ik}$\;

\For{$\ell=1, \dots, |\mathcal{K}(t)|$}{
    form $\widehat{\mathcal{J}}_{\pi(\ell)}^{(\ell-1)}$ with the $n_{\pi(\ell)}$ helpers having smallest candidate costs under the current state\;
    
    select $a_\ell^\star$ using \eqref{eq:heuristic_action_selection}\;
    
    apply $a_\ell^\star$ and update the state to $\mathbf{S}_{\ell}(t)$\;
}
\Return $\bm{\Psi}(t)$\;
\end{algorithm}

\subsubsection{Workflow}

For any state $\mathbf{S}_\ell$, let $\widehat{\mathcal{J}}_{ik}^{(\ell)}$ denote the set of $n_{ik}$ feasible helpers with smallest candidate costs $c_{ik}(\{j\} \mid \mathbf{S}_\ell)$. This candidate costs are only used to rank helpers individually; the final offloading decision is evaluated through the action cost $c_{ik}(\widehat{\mathcal{J}}_{ik}^{(\ell)}\mid\mathbf{S})$.

To prioritize tasks, the heuristic computes at the initial state $\mathbf{S}_0(t)$ the score
\begin{equation}
    \Delta_{ik}
    \defeq
    c_{ik}(\{i\}\mid\mathbf{S}_0(t))
    -
    c_{ik}(\widehat{\mathcal{J}}_{ik}^{(0)}\mid\mathbf{S}_0(t)).
    \label{eq:heuristic_priority_score}
\end{equation}
Tasks are then ordered in non-increasing order of $\Delta_{ik}$, so that those with larger offloading gain are processed first.

Starting from $\mathbf{S}_{0}(t)$, the heuristic processes the ordered tasks $\pi(1), \dots, \pi(|\mathcal{K}(t)|)$ sequentially. At step $\ell$, it forms $\widehat{\mathcal{J}}_{\pi(\ell)}^{(\ell-1)}$ under the current state $\ell-1$ and selects between local executions and offloading as
\begin{equation}
    a_\ell^\star
    \in
    \arg\min_{a\in
    \left\{
    \{\pi(\ell)_i\},
    \widehat{\mathcal{J}}_{\pi(\ell)}^{(\ell-1)}
    \right\}}
    c_{\pi(\ell)}(a\mid\mathbf{S}_{\ell-1}(t)).
    \label{eq:heuristic_action_selection}
\end{equation}

If the selected offloading set is infeasible, its cost is $+\infty$ by definition, and the local action is selected. After applying $a_\ell^\star$, the system state is updated to $\mathbf{S}_{\ell}(t)$ and the procedure continues with the next ordered task. At the end of the procedure, the final scheduling decision for time slot $t$ is the complete assignment $\bm{\Psi}(t)$ induced by the sequential decisions.

\subsubsection{Complexity analysis}

For each task, the priority score computation requires evaluating at most $N-1$ candidate helpers, followed by sorting the corresponding candidate scores, yielding $\mathcal{O}(|\mathcal{K}(t)| N \log N)$. During the sequential assignment stage, each ordered task again scans and ranks at most $N-1$ feasible helpers, which yields $\mathcal{O}(N \log N)$ per task. Therefore,
the overall per-slot time complexity is $\mathcal{O}\!\left(|\mathcal{K}(t)| \log |\mathcal{K}(t)| + |\mathcal{K}(t)| N \log N\right)$.

\section{Experimental Setup}
\label{sec:experimental_setup}

This section details the experimental setup used to evaluate the proposed coded task offloading scheme, including implementation details and the validation scenario. We also describe the considered comparative scenarios and the validation metrics used to analyze the performance of the proposed scheduling policies, including state-of-the-art solvers, common baselines, and other task-offloading schemes.

\subsection{Implementation Details}
\label{subsec:implementation_details}

To validate the proposed coded task offloading scheme, we implemented a Python-based discrete-event simulator (DES)\footnote{Available online: \url{https://gitlab.com/discovery1721326/d2d_task_offloading}.} The simulator follows the slot-based operation introduced in Section~\ref{sec:system_model} and evaluates the resulting delay, energy consumption, and privacy-leakage metrics under event-driven queue dynamics. In contrast to a purely analytical evaluation, the DES explicitly captures asynchronous execution, transmission and computation contention, queue concurrency, and memory-admission blocking. This event-driven execution is used as the default validation environment, since it captures runtime effects that are difficult to express in closed form. More details are shown in Appendix~\ref{app:more_implementation_details}.

\subsection{Validation Scenario}

We consider a heterogeneous D2D network where devices generate computation tasks locally and cooperate through short-range wireless links. To model different computational roles, devices are divided into three profiles: requesters, which generate tasks more frequently and have less energy-efficient computation; helpers, which generate fewer tasks and provide more efficient computation resources; and balanced devices, which represent an intermediate profile.

The default configuration is summarized in Table~\ref{tab:default_parameters}. The parameters are adapted from representative D2D and edge task-offloading studies~\cite{he2024multi,niu2025pipelining,baek2020privacy,ding2021potential,sufyan2020computation,zhao2017energy}. To avoid biasing the comparison through resource provisioning, all devices are assigned the same memory capacity, transmit power, and per-link bandwidth. Device heterogeneity is introduced through the task-arrival rate, CPU frequency, and effective switched capacitance coefficient. The communication model abstracts the D2D interface through an effective bandwidth of $250$ kHz per active link at $5.9$ GHz, which is consistent with NR sidelink/V2X operation over a fraction of the available resource pool~\cite{ETSI_TS_138_101_1_V16_12_1_2022}.

\begin{table}[tbh]
    \centering
    \caption{Default parameters for the validation experiments.}
    \label{tab:default_parameters}
    \renewcommand{\arraystretch}{1.1}
    \small
    \begin{tabularx}{\columnwidth}{@{} c Y @{}}
        \toprule
        \textbf{Notation} & \textbf{Default values} \\  \midrule \midrule
        $N$ & $50$ devices \\ \addlinespace
        $T,\tau$ & $T=300$, each of duration $\tau=0.1$ s \\ \addlinespace
        $M_i$ & $2$ Mbits \\ \addlinespace
        $(f_i,\kappa_i,\lambda_i)$ &
        \begin{tabular}[t]{@{}l@{}}
        requester ($50\%$): $(1.0~\mathrm{GHz},\,3.0\times 10^{-28},\,0.95)$ \\
        balanced ($30\%$): $(1.40~\mathrm{GHz},\,1.8\times 10^{-28},\,0.45)$ \\
        helper ($20\%$): $(1.6~\mathrm{GHz},\,1.5\times 10^{-28},\,0.05)$
        \end{tabular}
        \\ \addlinespace
        $p_{ij}(t)$ & $80$ mW $ \approx 19.03\, \text{dBm}$ \\ \addlinespace
        $\ell_{ik}$ & $\{0.22,\,0.42,\,0.75\}$ Mbits with prob. $(0.25,\,0.50,\,0.25)$ \\ \addlinespace
        $c_{ik}$ & $\{700,\,1000,\,1400\}$ cycles/bit with prob. $(0.30,\,0.45,\,0.25)$ \\ \addlinespace
        $B_{ij}(t)$ & $250$ kHz per active D2D link \\ \addlinespace
        $(n_{ik},s_{ik},t_{ik})$ & $\{(3,2,3), (4,3,4), (5,4,5), (6,3,5)\}$ and $q=256$, with equal probability \\ \addlinespace
        $N_0$ & $-174$ dBm/Hz \\ \addlinespace
        $f_\text{c}$ & $5.9$ GHz \\ \addlinespace
        $d_{ij}(t)$ & $\mathrm{U}[8,150]$ m \\ \addlinespace
        $\delta_j$ & $\{2^{-8},\,2^{-6},\,2^{-4}\}$ with equal probability \\ \addlinespace
        $\alpha,\beta,\gamma$ & $0.5, 0.5, 0.02$ \\ \addlinespace
        $D_0, E_0$ & $0.5\,\text{s}$, $0.2\,\text{J}$ \\ \addlinespace
        $\xi$ & $10$ \\
        \bottomrule
    \end{tabularx}
\end{table}

\subsection{Comparative Scenarios}
\label{subsec:comparative_scenarios}

\begin{figure*}[tbh]
    \centering
    \includegraphics[width=1.0\linewidth]{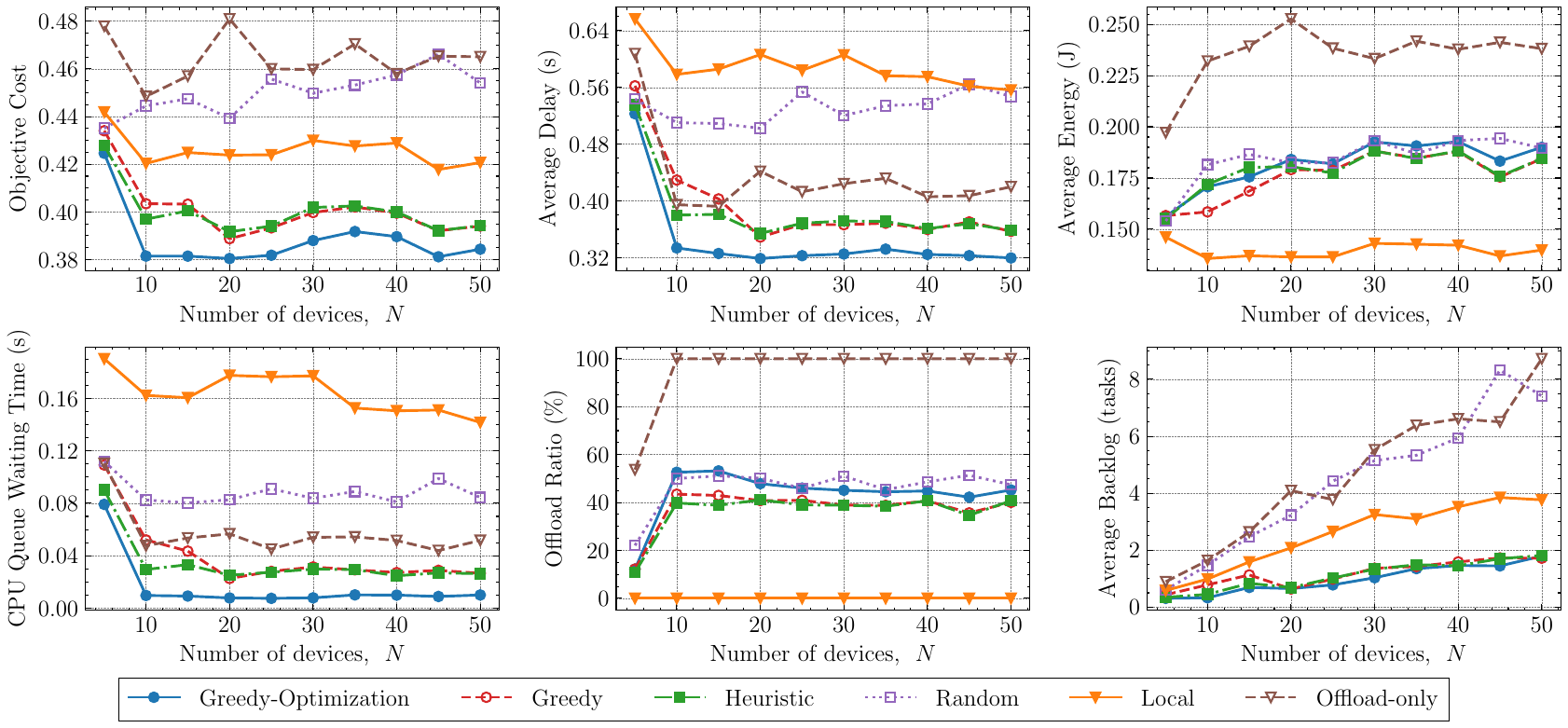}
    \caption{System performance comparison between the proposed solvers and common baselines under $\gamma=0$.}
    \label{fig:baseline_comparison}
\end{figure*}

We evaluate the proposed scheme from two complementary perspectives: the task offloading scheme used to execute each task, and the scheduling policy used to select the scheduling decision and helper devices. These are the considered task-offloading schemes:

\begin{itemize}
    \item \textbf{Full task offloading}: each task is either executed locally or offloaded to one helper device. This represents the conventional binary offloading case and is obtained from our model by disabling coding and setting $s_{ik}=n_{ik}=t_{ik}=1$.
    \item \textbf{Parallel task offloading}: each task is divided into several subtasks and assigned to multiple helpers. The task is completed only when all subtasks finish, which corresponds to disabling coded recovery and setting $s_{ik}=n_{ik}=t_{ik}=s$, with $s\in\{3,4,5,6\}$ equiprobable.
    \item \textbf{Coded task offloading}: this is the proposed scheme and the default configuration in this paper. Each task is encoded into $n_{ik}$ shares and can be recovered once the fastest subset of $t_{ik}$ helpers completes execution.
\end{itemize}

Moreover, we consider the following scheduling policies:
\begin{itemize}
    \item \textbf{Greedy-Optimization}: executes Algorithm~\ref{alg:reduced_bnb} over the recursive optimization model based on FCFS-queue dynamics, rather than over the DES, and is used as a near-optimal but optimistic reference.
    \item \textbf{Greedy}: executes Algorithm~\ref{alg:reduced_bnb} in the DES, allowing us to quantify the mismatch between the analytical model and the event-driven implementation.
    \item \textbf{Heuristic}: executes Algorithm~\ref{alg:heuristic_scheduler} in the DES. This is the default scheduler used throughout the paper.
    \item \textbf{Random}: randomly selects the execution mode and helper devices, when feasible.
    \item \textbf{Local}: forces all tasks to be executed locally.
    \item \textbf{Offload-only}: forces all offloadable tasks to be assigned to helper devices whenever feasible.
\end{itemize}

Finally, we adapt two state-of-the-art policies as representative baselines. We did not find in the state of the art any equivalent task offloading scheme under similar leakage penalty considerations. Therefore, for a fair comparison, we consider SEC2D~\cite{li2021security} and SMUA~\cite{malik2022efficient} under non-coded full and parallel task offloading schemes, respectively. Both algorithms are based on different techniques, namely Lyapunov optimization and matching theory, which provides an additional comparison for the proposed algorithmic solution. Moreover, each proposal focuses on a different main objective: energy consumption and queue stability in SEC2D, and task delay in SMUA.

\subsection{Performance Metrics}
\label{subsec:performance_metrics}

Most of the evaluated metrics are directly obtained from the closed-form models introduced in Sections~\ref{sec:system_model} and~\ref{sec:methodology}, including task delay, energy consumption, privacy penalty, objective cost, CPU/TX waiting time, and backlogged tasks. In addition, we report the average offload ratio, defined as the per-slot fraction of requested tasks that are not executed locally:
\begin{equation}
    \dfrac{1}{T}\sum_{t=1}^T\dfrac{|\{(i,k) \in \mathcal{K}(t) \,:\; \psi_{ik}^i = 0\}|}{|\mathcal{K}(t)|}
\end{equation}

\section{Results and Discussion}
\label{sec:results}

In this section, we validate the overall performance of the proposed coded task offloading scheme and solvers. First, we evaluate the general system performance, comparing the proposed solvers with baseline policies (Subsection~\ref{subsec:system_performance}). Next, we study the impact of the coding configuration and privacy-leakage penalty on the delay--energy trade-off under different noisy observations (Subsection~\ref{subsec:privacy_impact}). Finally, we compare the heuristic solver with state-of-the-art scheduling policies and alternative task-offloading schemes (Subsection~\ref{subsec:comparatives}).

\subsection{System Performance Validation}
\label{subsec:system_performance}

\subsubsection{Baseline Comparative}

Figure~\ref{fig:baseline_comparison} evaluates different metrics, including objective cost, average delay, energy consumption, CPU queue waiting time, offload ratio, and average backlog, under different scheduling policies. The main analysis comes from the objective cost, which is the metric minimized by the proposed solvers. Greedy-Optimization always obtains better results than Greedy, which is justified by the optimistic memory management described in Section~\ref{subsec:comparative_scenarios}. In contrast, there is no substantial difference between the objective cost achieved by Heuristic and Greedy. In fact, for small values of $N$, the Heuristic solver achieves slightly better results. This is because the optimization problem is solved sequentially: an assignment that is optimal in the current time slot may lead to a less favorable state in the following slots. For the remaining values of $N$, the Heuristic solver obtains very close results to Greedy, which shows that this lightweight approach can approximate near-optimal scheduling decisions. For instance, at $N=50$, the objective cost of Heuristic is $0.3943$, while Greedy obtains $0.3944$. The worst result is achieved by Offload-only due to the overhead introduced by the coded task offloading scheme.

The previous analysis over the objective cost can be similarly translated to task delay, CPU queue waiting time, and average backlog, where similar patterns are observed. However, this does not occur for task energy consumption. In this metric, Local achieves the lowest value because it avoids sequential share transmissions and replicated/coded remote computations. As a consequence, Local also achieves the worst task delay, since it cannot exploit helper devices to reduce the computation time of backlogged tasks.

Finally, the offload ratio is higher under Greedy-Optimization, $45.23\,\%$ for $N=50$, due to its optimistic resource management. In contrast, both Greedy and Heuristic require fewer offloads to achieve near-optimal solutions, around $40\,\%$. Remark that, the Offload-only solver does not reach a full offload ratio for $N=5$ due to the limited available helpers.

\subsubsection{Scheduling Time}

Table~\ref{tab:scheduling_runtime} shows the average per-slot scheduling runtime for different scheduling policies. Two main conclusions can be extracted. First, Greedy-Optimization requires longer runtime than Greedy because deriving waiting times and checking resource availability are more efficiently handled by DES. Second, Heuristic reduces the scheduling runtime by several orders of magnitude with respect to Greedy because it does not rely on the BnB procedure, while still achieving similar performance, as shown in the previous subsection. For instance, the mean runtime decreases from $2,784.40\times10^{-4}\,\text{s}$ with Greedy to $4.40\times10^{-4}\,\text{s}$ with Heuristic. On the other hand, the baseline policies are based on direct assignments, so their runtime is negligible compared to the optimization-based solvers. Offload-only requires slightly more runtime than Local and Random because it must check the current queue state to verify whether enough resources are available to offload the shares.

\begin{table}[t]
    \centering
    \caption{Per-slot scheduling runtime comparison. Values are reported in $10^{-4}$ s.}
    \label{tab:scheduling_runtime}
    \small
    \setlength{\tabcolsep}{2pt}
    \renewcommand{\arraystretch}{1.1}
    \begin{tabularx}{\columnwidth}{@{}Y c c c c c@{}}
        \toprule
        \textbf{Method} & \textbf{$N=5$} & \textbf{$N=10$} & \textbf{$N=20$} & \textbf{$N=50$} & \textbf{Mean} \\
        \midrule \midrule
        Greedy-Opt. & 2.00 & 170.30 & 1,736.80 & 10,963.80 & 3,847.80 \\
        Greedy & 1.90 & 132.20 & 680.00 & 7,262.90 & 2,784.40 \\
        Heuristic & 0.837 & 2.10 & 2.70 & 9.10 & 4.40 \\
        Random & 0.078 & 0.113 & 0.133 & 0.327 & 0.179 \\
        Local & 0.0178 & 0.0179 & 0.0189 & 0.0206 & 0.0178 \\
        Offload-only & 0.102 & 0.201 & 0.232 & 0.359 & 0.254 \\
        SEC2D~\cite{li2021security} & 0.437 & 0.813 & 1.40 & 4.00 & 2.00 \\
        SMUA~\cite{malik2022efficient} & 0.918 & 1.70 & 4.70 & 13.20 & 7.30 \\
        \bottomrule
\end{tabularx}
\end{table}

\subsubsection{System Scalability}

We also validated the system performance under different scalability conditions, including task-arrival regimes, workload variability levels, and time-slot durations. These experiments confirm that the proposed scheme remains stable under heavier arrivals and heterogeneous workloads, while adapting the offload ratio to queue accumulation and execution variability. Since these results mainly provide complementary evidence to the main conclusions, the detailed scalability analysis is reported in Appendix~\ref{app:system_scalability}.

\subsection{Privacy Impact on Delay-Energy Trade-off}
\label{subsec:privacy_impact}

Several experiments have been performed to analyze the impact of the considered privacy leakage model on the proposed task-offloading system. In this subsection, to enable comparable and controlled results, we assume a homogeneous noisy-observations leakage setting, i.e., $\delta_j=\delta$ for all helper devices.

\begin{figure}[tbh]
    \centering
    \includegraphics[width=1.0\linewidth]{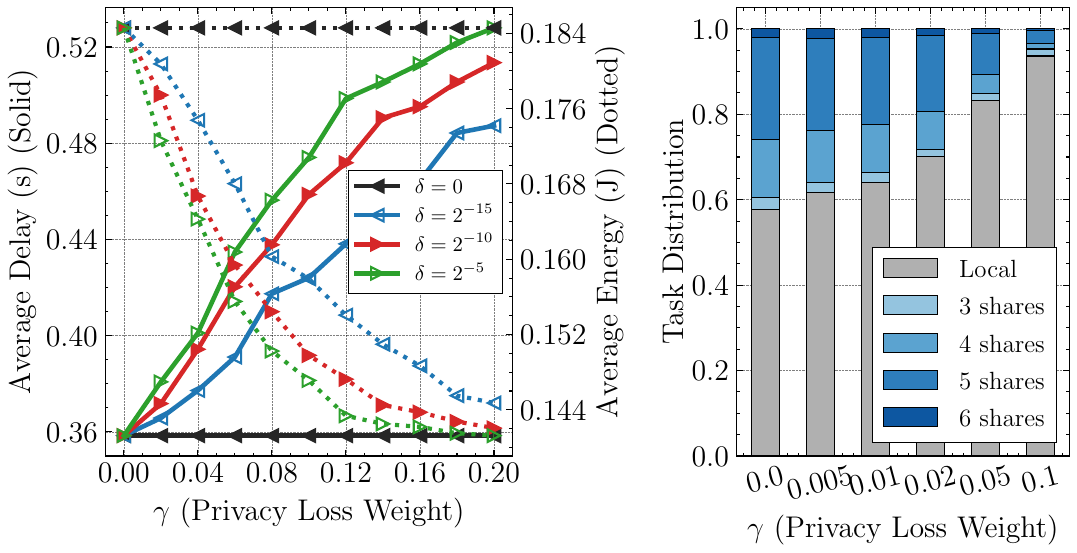}
    \caption{Impact of the privacy-leakage penalty on system performance under different values of $\gamma$.}
    \label{fig:privacy_gamma}
\end{figure}

\begin{figure*}[tbh]
    \centering
    \includegraphics[width=1.0\linewidth]{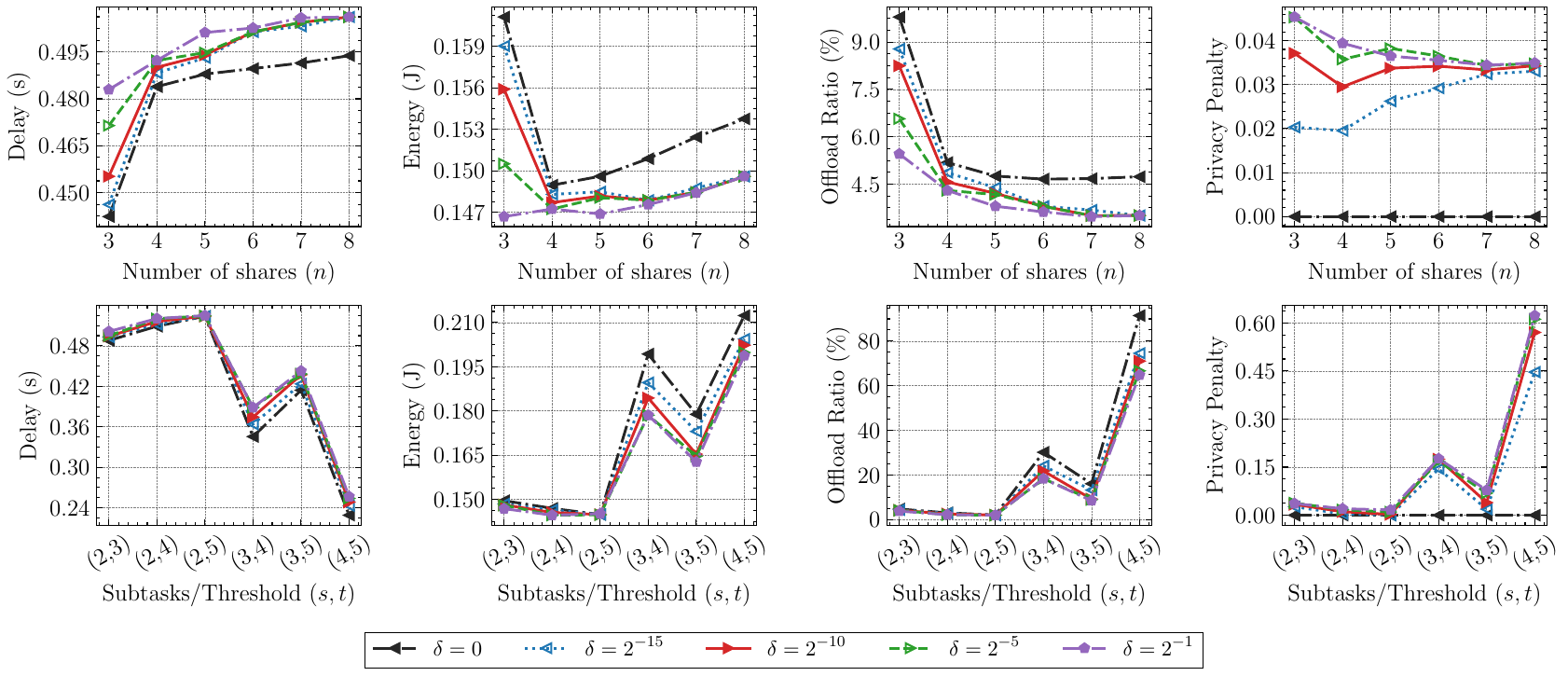}
    \caption{System performance under different coding parameters (in first row, $s_{ik}=2,t_{ik}=3$; in the second one, $n_{ik}=5$).}
    \label{fig:privacy_coding}
\end{figure*}

Figure~\ref{fig:privacy_gamma} studies the impact of the privacy penalty by modifying the control parameter $\gamma$, i.e., analyzing the delay--energy--privacy trade-off of the system. On the left side, both average delay and energy consumption remain constant for $\delta=0$, since no privacy leakage is introduced. Moreover, increasing $\gamma$ increases the average task delay, since the offload ratio decreases under the same setting. When the privacy penalty is further increased through the $\delta$-noisy parameter, the delay is penalized even more. The opposite occurs for energy consumption as higher values of $\gamma$ reduce the offload ratio to limit the objective cost, which also reduces the number of share transmissions and remote computations. Therefore, in coded task offloading under privacy-leakage penalty, task delay and privacy protection follow an inverse relationship (improving one degrades the other), and the same occurs between task delay and energy consumption; in contrast, energy consumption and privacy protection follow a direct relationship, since reducing offloading improves both metrics.  

On the right side, we focus on the distribution of scheduling decisions according to the number of shares used to encode each task. As expected from the previous results, increasing the privacy penalty through higher values of $\gamma$ increases the number of locally executed tasks, going from $57.8\%$ under the non-leakage scenario to $93.6\%$ for $\gamma=0.1$. A more detailed analysis can be derived by observing the number of shares, which ranges from $3$ to $6$ in our experimental setting. For $\gamma=0.0$, most offloaded tasks use $5$ shares, representing more than $20\%$ of the tasks, followed by $4$, $3$, and $6$ shares. Increasing $\gamma$ mainly penalizes the most frequent offloaded configurations, while the configurations with $3$ and $6$ shares keep, in relative terms, a similar share of offloaded tasks compared with $\gamma=0$. This behavior is explained by the overhead factor $n/s$ of the coding scheme, according to the coding parameters in Table~\ref{tab:default_parameters}. The configuration with the lowest overhead corresponds to $n=5$, while the configuration with the highest overhead corresponds to $n=7$. Therefore, the Heuristic solver tends to offload tasks with lower coding overhead. This also highlights an additional privacy trade-off as tasks with lower overhead are preferred for offloading, but they also correspond to less protective ramp secret sharing configurations.

This trade-off is further analyzed in Figure~\ref{fig:privacy_coding}. The same experiment is applied to two different scenarios: i) varying the number of shares under the same threshold-recovery and subtask configuration; ii) keeping the number of shares fixed to $5$ and varying the subtask and threshold-recovery profiles. By direct comparison, the latter produces a stronger impact on the overall system performance. In the former, increasing the number of shares generally makes offloading less attractive, since it requires additional transmission time and higher energy consumption at the helper nodes, while the number of completed shares required to finish the task remains unchanged. In any case, the offload ratio remains low, so its impact on the privacy penalty is negligible. In the latter, for a fixed number of shares, offloading becomes more attractive as the overhead factor decreases, i.e., when the number of subtasks increases, in this case. This is observed when moving from $(2,5)$ to $(4,5)$: even though the recovery threshold is the same, the offload ratio increases from below $5\%$ to approximately $70\%$. Additionally, when the number of subtasks is fixed, the offload ratio decreases as more completed shares are required to recover the task, as expected. Finally, the analysis over different values of $\delta$ follows the conclusions derived from the previously analyzed delay-energy-privacy trade-off .

\subsection{Comparison with other Works}
\label{subsec:comparatives}

\subsubsection{Task Offloading Schemes}

\begin{figure*}[tbh]
    \centering
    \includegraphics[width=1.0\linewidth]{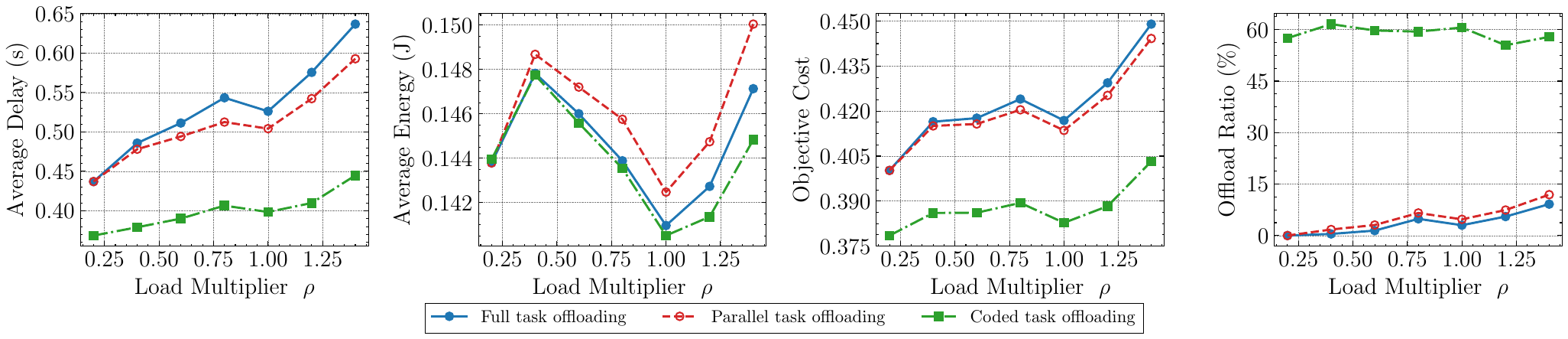}
    \caption{Comparison between task-offloading schemes under different task-arrival regimes.}
    \label{fig:comparative_offloading_schemes}
\end{figure*}

Figure~\ref{fig:comparative_offloading_schemes} illustrates a systematic comparison between several task offloading schemes, namely full, parallel, and coded task offloading, under different task-arrival regimes, obtained by scaling the default task-arrival profile by a factor $\rho\geq 0$. For a fair comparison, the leakage penalty is disabled by setting $\gamma=0$. Moreover, we set $B_{ij}(t)=500\,\text{kHz}$ to better highlight the differences between coded and non-coded schemes.

The coded approach achieves the lowest objective cost among the evaluated schemes, followed by parallel task offloading and full task offloading, which obtain close results. This behavior is mainly explained by the task delay, where coded offloading also achieves the best result due to straggler mitigation, and offload ratio. In coded offloading, offload ratio remains stable across the evaluated load factors $\rho$, around $60\,\%$, while full and parallel schemes remain below $15\,\%$. Similarly, coded offloading obtains the best energy-consumption result, closely followed by full offloading (which avoids the additional transmissions and remote computations introduced by task splitting) and parallel offloading (that must complete all subtasks before recovering the task). These conclusions apply to all evaluated task-arrival profiles, i.e., the effectiveness of each task-offloading scheme does not depend strongly on the current system workload.

\subsubsection{State-of-the-art Proposals}

\begin{figure*}[!t]
    \centering
    \begin{subfigure}{1.0\linewidth}
        \centering
        \includegraphics[width=1.0\linewidth]{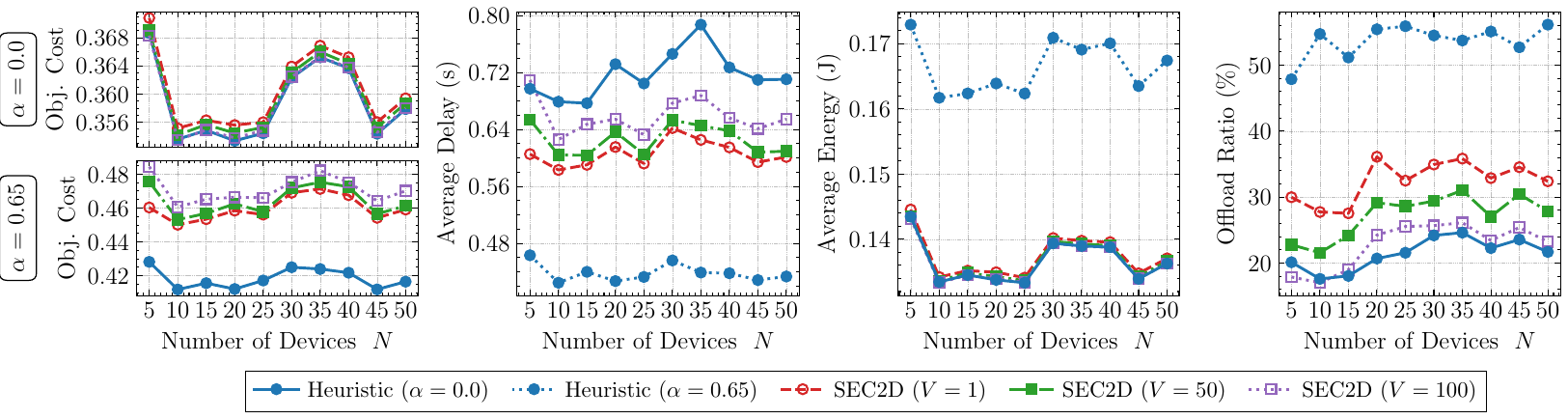}
        \caption{Comparison between SEC2D~\cite{li2021security} and the proposed Heuristic solver under full task offloading.}
        \label{fig:sota_comparative_a}
    \end{subfigure}
    \vspace{0.3cm}
    \begin{subfigure}{1.0\linewidth}
        \centering
        \includegraphics[width=1.0\linewidth]{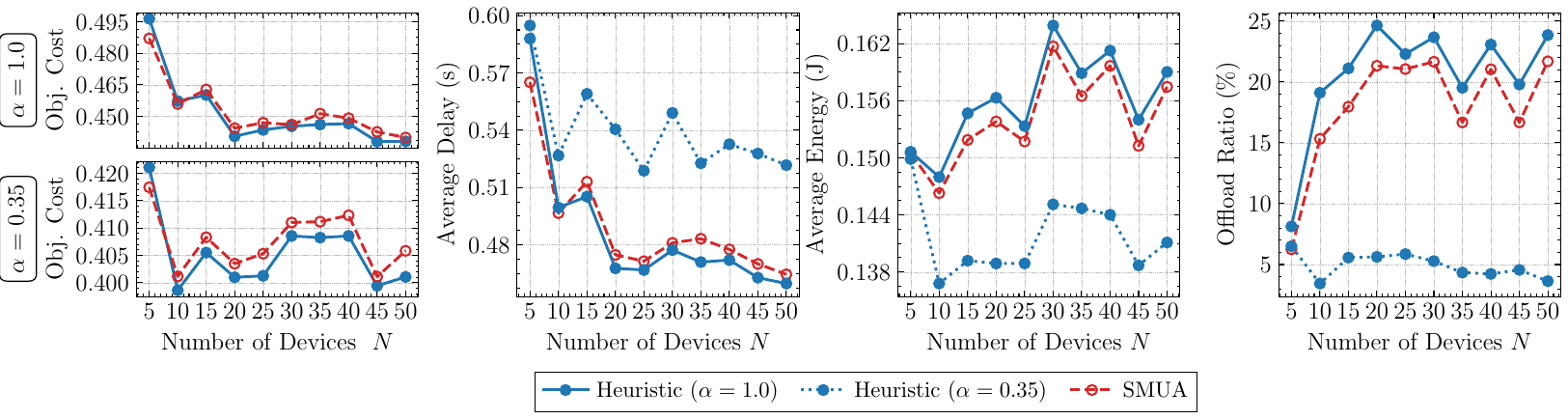}
        \caption{Comparison between SMUA~\cite{malik2022efficient} and the proposed Heuristic solver under parallel task offloading.}
        \label{fig:sota_comparative_b}
    \end{subfigure}
    \caption{Comparison between representative state-of-the-art task-offloading schemes and the proposed Heuristic solver under full and parallel task offloading settings.}
    \label{fig:sota_comparative}
\end{figure*}

In terms of scheduling policy effectiveness, we compare the proposed Heuristic solver with the selected state-of-the-art policies in Section~\ref{subsec:comparative_scenarios}. Since no directly equivalent privacy-aware offloading scheme is available, we compare the proposed Heuristic solver with SEC2D and SMUA under full and parallel offloading, respectively, with the privacy penalty disabled $(\gamma=0)$. These baselines provide representative comparisons against Lyapunov- and matching-based approaches, while varying the number of devices to assess scalability.

Figure~\ref{fig:sota_comparative_a} illustrates the validated metrics under the Heuristic solver and several executions of SEC2D, each corresponding to a different value of $V$, a parameter that controls the energy-aware Lyapunov policy under queue stability. When $V\to0$, SEC2D prioritizes queue stability, whereas when $V\to+\infty$, it prioritizes energy consumption. This experiment is evaluated for $\alpha=0.0$, which is equivalent to the scenario considered in~\cite{li2021security}, and for $\alpha=0.65$ to observe the delay--energy trade-off. For $\alpha=0.0$, the Heuristic solver achieves the best result in terms of objective cost, equivalently energy consumption, closely followed by most SEC2D configurations. In fact, for $V=100$, SEC2D achieves close results for most values of $N$. However, the Heuristic solver obtains these results with fewer offloaded tasks, as shown in the figure, which indicates a more selective helper assignment. As a consequence, the Heuristic solver penalizes task delay and obtains the worst delay result among the evaluated scenarios.

When shifting the value of $\alpha$, i.e., adding priority to task delay, the Heuristic solver clearly outperforms SEC2D in terms of objective cost. Remark that the results of SEC2D under this scenario remain the same as before, except for the objective cost, because its scheduling decision is not directly controlled by $\alpha$. In this case, the Heuristic solver reduces task delay by increasing the number of offloaded tasks. The price is a clear increase in energy consumption, which is consistent with the delay--energy trade-off analyzed throughout this section.

Additionally, according to Table~\ref{tab:scheduling_runtime}, SEC2D requires lower scheduling runtime than the Heuristic solver, with mean values of $2.00\times10^{-4}\,\text{s}$ and $4.40\times10^{-4}\,\text{s}$, respectively. However, this lower runtime comes with a more rigid optimization structure. Since SEC2D is driven by a Lyapunov-based energy--queue stability formulation, its control parameter mainly regulates the trade-off between energy consumption and queue stability, but it does not directly adapt to the delay--energy objective used in this work.

On the other hand, Figure~\ref{fig:sota_comparative_b} illustrates the validated metrics under the Heuristic solver and SMUA. This experiment is evaluated for $\alpha=1.0$, which is equivalent to the scenario considered in~\cite{malik2022efficient}, and for $\alpha=0.35$ to observe the delay--energy trade-off. For $\alpha=1.0$, the proposed solver achieves the best results in terms of objective cost, equivalently task delay, for most values of $N$, especially when $N>10$. These results are achieved because the Heuristic solver decides to offload more tasks than SMUA in all evaluated scenarios. The price to pay is energy consumption, which is slightly penalized by the Heuristic solver.

When shifting the value of $\alpha$, i.e., adding priority to energy consumption, the Heuristic solver improves its previous result by substantially reducing the energy consumption with respect to SMUA. This improvement is achieved by decreasing the number of offloaded tasks. As previously analyzed, task delay is then penalized as a direct consequence of the lower offload ratio.

Additionally, according to Table~\ref{tab:scheduling_runtime}, matching-based SMUA is slower than the proposed Heuristic solver, with mean values of $7.30\times10^{-4}\,\text{s}$ and $4.40\times10^{-4}\,\text{s}$, respectively. This occurs because SMUA requires building and updating preference relations between task subtasks and helper nodes, and the matching process must account for changes in the allocation state. 

\section{Conclusions and Future Work}
\label{sec:conclusions}

This paper introduced a privacy-aware coded task offloading scheme under D2D queue-based dynamics and stochastic task arrivals. for optimizing distributed application deployments. The proposal combines the advantages of task offloading and coded computing by providing straggler resistance through threshold-based recovery, improving the delay--energy trade-off through a cancellation-aware execution mechanism. Additionally, the privacy-aware system is modeled through a penalization on the information leaked by offloaded coded shares under a noisy leakage observations (such as side-channel attacks), by quantifying the theoretical privacy leakage of the underlying LSS scheme. Therefore, the proposed model does not only decide where tasks are executed, but also evaluates how coded execution affects delay, energy consumption, and privacy leakage.

Evaluation tests validated the performance of the proposed BnB-based greedy and heuristic solvers through the designed discrete-event simulator. The results provide three complementary points of view: i) a system-level analysis assessing scalability, task arrival regimes, and workload variability; ii) an analysis of the impact of the privacy penalty on the delay--energy trade-off; and iii) a comparison of the proposed solvers under alternative task-offloading schemes, baselines, and state-of-the-art scheduling policies. In general, the proposed coded task offloading scheme outperforms classical full and parallel task offloading schemes, while the heuristic solver achieves near-greedy performance with significantly lower computational cost. The results also show a No-Free-Lunch-like behavior: under coded task offloading, delay, energy consumption, and information-based privacy leakage cannot be jointly optimized. Thus, distributed service deployments must explicitly prioritize the most relevant requirements depending on the target application.

Further studies are still needed to enhance resource allocation and workload placement in Fluid Computing environments. On the one hand, the proposed solvers provide effective centralized and per-slot decisions that allow us to validate the privacy-aware coded offloading model and quantify its delay--energy trade-off. However, highly dynamic and massive networks require more scalable and distributed algorithmic solutions, able to operate with partial information, adapt to changing D2D neighborhoods, and preserve low scheduling overhead. On the other hand, while this work focuses on the privacy implications of coded task offloading, Fluid Computing also requires higher-level orchestration mechanisms capable of coordinating offloading decisions across heterogeneous resources and administrative domains. Therefore, future work will extend this line toward distributed and privacy-aware task offloading policies integrated into fluid orchestration frameworks, where coded execution can be selected at runtime according to application intents, resource availability, mobility, and multi-domain settings.

\section*{Acknowledgment}
This work was supported by the grant PID2023-148716OB-C31 funded by MCIU/AEI/10.13039/501100011033 (DISCOVERY project). Additionally, it has also been funded by the Galician Regional Government under project ED431B 2024/41 (GPC).

\input{./refs.bbl}

\newpage
\clearpage
\appendix

\section{Derivation of the Privacy Leakage Bound}
\label{app:privacy_leakage_derivation}

\subsection{Auxiliary Results and Adapted Leakage Framework}

We follow the Fourier-analytic leakage derivation of~\cite{gupta2026security} and adapt it to the multi-secret ramp secret sharing scheme introduced in Section~\ref{subsec:privacy_leakage_model}. To derive the leakage bound, we recall the required dual-code notation and establish the uniformity properties of the generated shares.

For any linear code $\mathcal{C} \subseteq \mathbb{F}_{q}^{n}$, its dual code is
\[
  \mathcal{C}^{\perp} = \{\mathbf{x} \in \mathbb{F}_{q}^{n} :
  \langle\mathbf{x},\mathbf{c}\rangle=0, \ \forall\mathbf{c}\in\mathcal{C}\}.
\]
Since $\mathcal{C}_{2} \subsetneq \mathcal{C}_{1}$, we have $\mathcal{C}_{1}^{\perp}\subsetneq\mathcal{C}_{2}^{\perp}$. Moreover, because the dual of an MDS code is also MDS, $\mathcal{C}_{1}^{\perp}$ is an $[n,n-t,t+1]_{q}$ code, and $\mathcal{C}_{2}^{\perp}$ is an $[n,n-t+k,t-k+1]_{q}$. Therefore, every non-zero vector in $\mathcal{C}_{2}^{\perp}$ has Hamming weight at least $t-k+1$.
 
The following simple lemma shows that the secret shares in an $\mathrm{RampSS}(t,k,n)$ are uniformly distributed over $\mathbb{F}_q$.

\begin{lemma}
  Consider a $\mathrm{RampSS}(t,k,n)$ scheme instantiated with nested MDS codes. Assume that $\mathbf{S}$ is uniformly distributed over $\mathbb{F}_{q}^{k}$. Then, $U_i$ is uniformly distributed in $\mathbb{F}_q$, for every $i = 1, \dots, n$.
\end{lemma}
\begin{proof}
  Consider the mapping $\Phi:\mathbb{F}_{q}^{s}\times\mathcal{C}_{2}\to\mathcal{C}_{1}$ with $\Phi(\mathbf{s},\mathbf{c}_{2})=\psi(\mathbf{s})+\mathbf{c}_{2}$. Since $\mathcal{C}_{1}=\mathcal{S}\oplus\mathcal{C}_{2}$ and $\psi:\mathbb{F}_{q}^{s}\to\mathcal{S}$ is an isomorphism, $\Phi$ is bijective (i.e., every $\mathbf{c}_1 \in \mathcal{C}_1$ has a unique preimage $(\mathbf{s}, \mathbf{c}_2)$). Moreover, $\mathbf{S}$ is uniform over $\mathbb{F}_{q}^{k}$ and the random vector inside each selected coset is uniform over $\mathcal{C}_{2}$. Therefore, $\mathbf{U}$ is uniformly distributed over $\mathcal{C}_{1}$.

  Now fix a coordinate $i$. Let $\pi_{i}:\mathcal{C}_{1}\to\mathbb{F}_{q}$ be the linear mapping defined by $\pi_{i}(\mathbf{c})=c_{i}$. Since $\mathcal{C}_{1}$ is an MDS code with parameters $[n,t,n-t+1]_{q}$, no coordinate projection of $\mathcal{C}_{1}$ is identically zero. Otherwise, puncturing that coordinate would give a length-$(n-1)$ linear code of dimension $t$ and minimum distance $n-t+1$, contradicting the Singleton bound. Therefore, $\pi_{i}$ is a non-zero linear form and hence is surjective onto $\mathbb{F}_{q}$. By standard algebra, $\mathcal{C}_{1}/\operatorname{ker}(\pi_{i})$ is isomorphic to $\mathbb{F}_{q}$, and the cosets $\{\mathbf{c}_{\alpha}+\operatorname{ker}(\pi_{i}): \alpha\in\mathbb{F}_{q},\ \pi_{i}(\mathbf{c}_{\alpha})=\alpha\}$ are a partition of $\mathcal{C}_{1}$. Since $\mathbf{U}$ is uniformly distributed over $\mathcal{C}_{1}$, it follows that $U_{i}=\pi_{i}(\mathbf{U})$ is uniformly distributed in $\mathbb{F}_{q}$.
\end{proof}

On the other hand, in preparation of our main result, let us repeat here the definitions
introduced in~\cite{gupta2026security}, adapted to our setting.
\begin{definition}
  Given a secret vector $\mathbf{s}$ and some leakage realization
  $\boldsymbol\ell$, define the information density
  \begin{equation*}
    \Lambda(\mathbf{s}, \boldsymbol\ell) := \frac{\mathbb{P}(\mathbf{S} =
      \mathbf{s}, \mathbf{L} = \boldsymbol\ell)}{\mathbb{P}(\mathbf{S} =
      \mathbf{s}) \mathbb{P}(\mathbf{L} = \boldsymbol\ell)},
  \end{equation*}
  and the information density under independence
  \begin{equation*}
    \Lambda^\prime(\mathbf{s}, \boldsymbol\ell) := \frac{\mathbb{P}(\mathbf{S} =
      \mathbf{s}, \mathbf{L} = \boldsymbol\ell)}{\mathbb{P}(\mathbf{S} =
      \mathbf{s}) \prod_{j = 1}^n \mathbb{P}(L_i = \ell_i)}.
  \end{equation*}
\end{definition}
The name \emph{information density} is justified by the simple observation
that
\begin{equation*}
  I(\mathbf{S}; \mathbf{L}) = \mathbb{E} \bigl[ \log \Lambda(\mathbf{S},
  \mathbf{L}) \bigr],
\end{equation*}
where the expectation is over the joint distribution of
$(\mathbf{S}, \mathbf{L})$ and $I(\mathbf{S}; \mathbf{L})$ is the mutual
information between two random vectors. In a similar spirit,
$\Lambda^\prime$ captures the statistics that compare the channel
$\mathbf{S} \mapsto \mathbf{L}$ with the separate random variables
$\mathbf{S}$ and the iid vector $(L_1, \dots, L_n)$. Through the lens of
information theory, this quantifies to what extent it is possible to simulate
the channel $\mathbf{S} \mapsto \mathbf{L}$ with independent random variables,
or equivalently if the two alternatives can be statistically
distinguished. Moreover, an important identity is $2 \Delta^{\text{TV}} = \mathbb{E} \bigl[ | \Lambda - 1 | \bigr]$, where again the expectation is with respect to the joint distribution of
$(\mathbf{S}, \mathbf{L})$.

\begin{lemma}
  Given a secret vector $\mathbf{s}$ and a leakage realization
  $\boldsymbol\ell$
  \begin{equation*}
    \left| \frac{\Lambda^\prime(\mathbf{s},
        \boldsymbol\ell)}{\Lambda(\mathbf{s}, \boldsymbol\ell)} - 1 \right|
    \leq \max_{\mathbf{s} \in \mathbb{F}_q^k} \bigl|
    \Lambda^\prime(\mathbf{s}, \boldsymbol\ell) - 1\bigr|.
  \end{equation*}
\end{lemma}
\begin{proof}
We have that
\begin{equation*} 
\begin{aligned} 
    \left|\frac{\Lambda'}{\Lambda}-1\right| 
    &= \left|\frac{\mathbb{P}(\mathbf{L}=\boldsymbol{\ell})}{\prod_{j=1}^{n}\mathbb{P}(L_j=\ell_j)}-1\right| \\ 
    &= \left|\mathbb{E}_{\mathbf{S}}\!\left[\Lambda'(\mathbf{S},\boldsymbol{\ell})-1\right]\right| \\ 
    &\leq \mathbb{E}_{\mathbf{S}}\!\left[\left|\Lambda'(\mathbf{S},\boldsymbol{\ell})-1\right|\right]. 
\end{aligned} 
\end{equation*}
and the lemma follows by using Jensen's inequality over convex functions.
\end{proof}
We can also bound the maximum deviation of the information density as follows.
\begin{lemma}
  Fix a secret $\mathbf{s}$ and a leakage realization
  $\boldsymbol\ell$. It holds that
  \begin{equation*}
    \max_{\mathbf{s} \in \mathbb{F}_q^k} \bigl| \Lambda(\mathbf{s},
    \boldsymbol\ell) - 1 \bigr| \leq 2 \frac{\Lambda(\mathbf{s},
      \boldsymbol\ell)}{\Lambda^\prime(\mathbf{s}, \boldsymbol\ell)} \left(
    \max_{\mathbf{s} \in \mathbb{F}_q^k} \bigl| \Lambda^\prime(\mathbf{s},
    \boldsymbol\ell) - 1 \bigr| \right) 
  \end{equation*}
\end{lemma}
\begin{proof}
This follows by applying the triangle inequality
\begin{equation*} 
\begin{aligned}
\left|\Lambda(\mathbf{s},\boldsymbol{\ell})-1\right| 
&\leq \left|\Lambda(\mathbf{s},\boldsymbol{\ell})-\frac{\Lambda}{\Lambda'}\right|+\left|\frac{\Lambda}{\Lambda'}-1\right| \\ 
&= \frac{\Lambda(\mathbf{s},\boldsymbol{\ell})}{\Lambda'(\mathbf{s},\boldsymbol{\ell})}\left|\Lambda'(\mathbf{s},\boldsymbol{\ell})-1\right|+\left|\frac{\Lambda(\mathbf{s},\boldsymbol{\ell})}{\Lambda'(\mathbf{s},\boldsymbol{\ell})}-1\right|. 
\end{aligned} 
\end{equation*}

  Reorganizing these terms
  \begin{equation*}
    \max_{\mathbf{s} \in \mathbb{F}_q^k} \bigl| \Lambda(\mathbf{s},
    \boldsymbol\ell) - 1 \bigr| \leq 2 \frac{\Lambda(\mathbf{s},
      \boldsymbol\ell)}{\Lambda^\prime(\mathbf{s}, \boldsymbol\ell)} \left(
      \max_{\mathbf{s} \in \mathbb{F}_q^k} \bigl| \Lambda^\prime(\mathbf{s},
      \boldsymbol\ell) - 1 \bigr| \right).
  \end{equation*}
  Note that for the latter inequality we have used the result in Lemma 2.
\end{proof}
We now address the question of bounding $\max_{\mathbf{s} \in \mathbb{F}_q^k}
| \Lambda^\prime(\mathbf{s}, \boldsymbol\ell) - 1|$. To that end, define first
the a posteriori distributions
\begin{align*}
  \Psi_{\ell_i}: \mathbb{F}_q &\rightarrow \mathbb{C} \\
  \alpha &\mapsto \mathbb{P}(U_i = \alpha \mid L_i = \ell_i),
\end{align*}
and let $\Psi_{\boldsymbol\ell}(\alpha_1, \dots, \alpha_n) := \prod_{j = 1}^n
\Psi_{\ell_j}(\alpha_j)$.\\

\begin{lemma}
    For a $\mathrm{RampSS}(t,k,n)$ scheme, let $\boldsymbol{\ell}$ be a fixed leakage realization under the $\delta$-noisy leakage model. Then,
  \begin{equation*}
    \max_{\mathbf{s} \in \mathbb{F}_q^k} \bigl| \Lambda^\prime(\mathbf{s},
    \boldsymbol\ell) - 1 \bigr| \leq \sum_{\mathbf{w} \in \mathcal{C}_2^{\perp,\times}} \bigl|
  \hat{\Psi}_{\boldsymbol\ell}(\mathbf{w}) \bigr|.
  \end{equation*}
\end{lemma}
\begin{proof}
  The Lemma and the proof are a special case of Lemma 4 in~\cite{gupta2026security}. Indeed, we can write that
  \begin{equation*}
  \begin{aligned}
    \mathbb{P}(\mathbf{L} = \boldsymbol\ell \mid \mathbf{S} = \mathbf{s}) 
    & = \dfrac{\sum_{\mathbf{u}: \mathbf{S} = \mathbf{s}} \mathbb{P}(\mathbf{U}=\mathbf{u}) \, \mathbb{P}(\mathbf{L} = \boldsymbol\ell \mid \mathbf{U}=\mathbf{u})}{\sum_{\mathbf{u}: \mathbf{S} = \mathbf{s}} \mathbb{P}(\mathbf{U}=\mathbf{u})}\\
    & = q^{k-t} \sum_{\mathbf{u}: \mathbf{S} = \mathbf{s}} \prod_{i=1}^n \mathbb{P}(L_i = \ell_i \mid U_i = u_i) \\
    & = q^{k+n-t} \sum_{\mathbf{u}: \mathbf{S} = \mathbf{s}} \prod_{i=1}^n  \mathbb{P}(L_i = \ell_i) \\
    & \hspace{5.2em}\cdot \mathbb{P}(U_i = u_i \mid L_i = \ell_i) \\
    &= q^{k+n-t} \prod_{i = 1}^n \mathbb{P}(L_i = \ell_i) \sum_{\mathbf{u} \in \psi(\mathbf{s})+\mathcal{C}_2}
    \Psi_{\boldsymbol\ell}(\mathbf{u})
  \end{aligned}
  \end{equation*}
  by using the Bayes's rule and since there exists a \emph{unique} coset of secrets that results in the same share vector $\mathbf{u}$. And by the Poisson summation formula (Theorem 2 from~\cite{gupta2026security})
  \begin{align*}
    \Lambda^\prime(\mathbf{s}, \boldsymbol\ell) 
    &= \frac{\mathbb{P}(\mathbf{L}
      = \boldsymbol\ell \mid \mathbf{S} = \mathbf{s})}{\prod_{i = 1}^n
      \mathbb{P}(L_i = \ell_i)} \\
      &= \dfrac{q^{k+n-t}}{|\mathcal{C}_2^\perp|} \cdot \sum_{\mathbf{u} \in \mathcal{C}_2^\perp} \hat{\Psi}_{\boldsymbol\ell}(\mathbf{u}) \chi(-\langle \psi(\mathbf{s}), \mathbf{w} \rangle)\\
      & = \sum_{\mathbf{w} \in \mathcal{C}_2^{\perp}}
    \hat{\Psi}_{\boldsymbol\ell}(\mathbf{w}) \chi(-\langle \psi(\mathbf{s}), \mathbf{w} \rangle),
  \end{align*}
  where $\chi: \, \mathbb{F}_q \to \mathbb{C}$ is the character function. Since $\hat{\Psi}_{\boldsymbol\ell}(0)\chi(0) = 1$ and $|\chi(\cdot)| = 1$,
  we get
  \begin{equation*}
    \bigl| \Lambda^\prime(\mathbf{s}, \boldsymbol\ell) - 1 \bigr| \leq
    \sum_{\mathbf{w} \in \mathcal{C}_2^{\perp,\times}} \bigl|
    \hat{\Psi}_{\boldsymbol\ell}(\mathbf{w}) \bigr|,
  \end{equation*}
  by applying the triangle inequality.
\end{proof}

\subsection{Proof of Theorem 1}

Recall that $|\hat{\Psi}_{\boldsymbol\ell}(\mathbf{w})| = \prod_{i =
  1}^n |\hat{\Psi}_{\ell_i}(w_i)|$, that $\bigl|  \hat{\Psi}_{\ell_i}(0)\bigr| = 1$, and that
\begin{equation*}
  \bigl| \hat{\Psi}_{\ell_i}(w_i) \bigr| \leq \operatorname{bias}(U_i \mid L_i =
  \ell_i) = 2 \Delta_{\ell_i}^{\text{FL}}(U_i, L_i).
\end{equation*}

Moreover, the Fourier coefficients $| \hat{\Psi}_{\ell_j}(w_j)|$ are smaller in magnitude than the trivial one $|\hat{\Psi}_{\ell_j}(0)| = 1$ as highlighted in~\cite{benhamouda2021local}. Consequently, it is sufficient to pick any $\mathbf{w}$ of minimum weight and bound the Fourier coefficients. Since the sum is performed over $\mathcal{C}_2^{\perp}$, we know that the minimum weight of non-zero codewords is $t-k+1$ due to its MDS structure.

To finalize the bounds on the leakage, we recall that
$\Delta^{\text{ARE}}(\mathbf{S}, \mathbf{L}) = \mathbb{E}_{\mathbf{L}} \bigl[
\max_{\mathbf{s} \in \mathbb{F}_q^k} \Lambda(\mathbf{s}, \mathbf{L}) - 1
\bigr]$. Therefore,
\begin{equation*} 
\begin{aligned} 
\frac{1}{2}\Delta^{\mathrm{ARE}}(\mathbf{S},\mathbf{L}) 
&\leq \mathbb{E}_{\mathbf{L}}\!\left[\frac{\Lambda(\mathbf{s},\boldsymbol{\ell})}{\Lambda'(\mathbf{s},\boldsymbol{\ell})}\max_{\mathbf{s}\in\mathbb{F}_{q}^{k}}\left|\Lambda'(\mathbf{s},\boldsymbol{\ell})-1\right|\right] \\ 
& = \sum_{\boldsymbol{\ell} \in \mathcal{L}} \prod_{i=1}^{n} \mathbb{P}(\mathbf{L} = \boldsymbol{\ell})  \max_{\mathbf{s}\in\mathbb{F}_{q}^{k}}\left|\Lambda'(\mathbf{s},\boldsymbol{\ell})-1\right|\\
&\leq \sum_{\mathbf{w}\in\mathcal{C}_{2}^{\perp,\times}}\prod_{i=1}^{n}\mathbb{E}_{L_i}\!\left[\left|\widehat{\Psi}_{\ell_i}(w_i)\right|\right]. 
\end{aligned} 
\end{equation*}
Using that
$\mathbb{E}_{L_j} \bigl[ \bigl| \hat{\Psi}_{\ell_j}(w_j) \bigr| \bigr]
\leq 2 \Delta^{\text{FL}}(U_j, L_j)$, and letting $\mathcal{J} \subseteq \text{supp}(\mathbf{w})$ be a subset of size $t-k+1$ among $[n]$, we get the bound
\begin{equation*}
\begin{aligned} 
  \frac{1}{2} \Delta^{\text{ARE}}(\mathbf{S}, \mathbf{L}) 
  & \leq \sum_{\mathbf{w}\in\mathcal{C}_{2}^{\perp,\times}} 2^{t-k+1} \prod_{j \in \mathcal{J}} \Delta^{\text{FL}}(U_j, L_j)\\
  & \leq 2^{t-k+1}
    \left(q^{n-t+k}-1\right)
    \max_{\mathcal{J} \in \binom{[n]}{t-k+1}}
    \prod_{j \in \mathcal{J}} \delta_{j},
\end{aligned} 
\end{equation*}
It is possible to bound also the mutual information,
\begin{equation*} 
\begin{aligned} 
I(\mathbf{S};\mathbf{L}) 
&\leq H(\mathbf{S})\,\Delta^{\mathrm{ARE}}(\mathbf{S},\mathbf{L}) \\ 
&\leq k\,2^{t-k+1}\left(q^{n-t+k}-1\right)\log_{2}q \\ 
&\hspace{1.4em} \cdot \max_{\mathcal{J}\in\binom{[n]}{t-k+1}}\prod_{j\in\mathcal{J}}\delta_j. 
\end{aligned} 
\end{equation*}

\section{NP-hardness of Problem $\mathcal{P}1$} 
\label{app:np_hardness}

\begin{proof}
We prove the claim by polynomial-time reduction from a classical NP-hard scheduling problem. Specifically, we consider the minimum mean finishing-time problem for independent tasks on parallel processors, shown in~\cite{bruno1974scheduling}, to be polynomial complete.

Consider a single time slot $t$ whose objective reduces to the per-slot average cost over the batch $\mathcal{K}(t)$. Then, impose the following restrictions:
\begin{enumerate}
    \item[\emph{a)}] the objective function in \eqref{eq:problem_obj} is restrained to a only delay-aware objective by choosing $\alpha=1$, and the privacy-leakage penalty is made inactive by choosing $\gamma=0$;
    \item[\emph{b)}] for every task $(i,k)$, set $s_{ik}=t_{ik}=n_{ik}=1$, so that offloading means assigning the overall task to exactly one helper;
    \item[\emph{c)}] the memory constraint \eqref{eq:problem_a} is made inactive by taking sufficiently large resources;
    \item[\emph{d)}] make local execution dominated by assigning it an arbitrarily large cost, so that every optimal solution offloads each task to one helper device;
    \item[\emph{e)}] suppress communication costs (both transmission time and transmission energy consumption), so that only computation queues remain relevant.
\end{enumerate}

Under these restrictions, every task $(i,k)\in\mathcal{K}(t)$ must be assigned to exactly one helper device $j$, and its delay is determined by the the completion time induced by the computation queue of device $j$ and its computation resources. Hence, the restricted problem becomes a scheduling problem for assigning independent jobs to a set of parallel resources. More precisely, let an instance of the source problem contain $n$ jobs and $m$ processors, and let $p_{qj}$ denote the processing time of job $q$ on processor $j$. We construct an instance of the restricted version of $\mathcal{P}1$ by mapping one task in $\mathcal{K}(t)$ to each job, one helper device for each processor, and setting the computation time of the corresponding task on helper $j$ equal to $p_{qj}$. This mapping requires writing $n$ tasks, $m$ helper devices, and the $n\times m$ matrix $\{p_{qj}\}$, and is therefore computable in polynomial time. Because communication is suppressed and local execution is not considered, scheduling decisions of the restricted $\mathcal{P}1$ are in one-to-one mapping with scheduling decisions of the source scheduling instance.

Since minimum mean completion-time scheduling on parallel processors is NP-hard, the restricted version of $\mathcal{P}1$ is NP-hard. Therefore, the full problem $\mathcal{P}1$ is NP-hard as well.
\end{proof}

\section{More Implementation Details}
\label{app:more_implementation_details}

The DES is implemented using the SimPy framework\footnote{\url{https://www.sympy.org/}} and validates the scheduling decisions in an event-driven queueing environment. Its workflow is summarized in Figure~\ref{fig:DES_scheme}. At each scheduling epoch, the incoming tasks \texttt{TaskBatch} and the current \texttt{SystemState} are provided to the scheduling policy, which returns a \texttt{SchedulePlan}. This plan is then executed by the Execution Engine, which instantiates the scheduled \texttt{Workload} and manages the \texttt{CodedTaskController} for tracking the completion of the encoded shares required to reconstruct each output. The Node Abstraction represents the heterogeneous devices and encapsulates their local resources, including \texttt{SimPy Resources} for serialized transmission and computation operations, and \texttt{SimPy Containers} for memory occupancy. Once the time-slot execution finishes, the simulator updates the system state, generates the next task batch, and repeats the process until the last slot is reached.

\begin{figure}[tbh]
    \centering
    \includegraphics[width=1.0\linewidth]{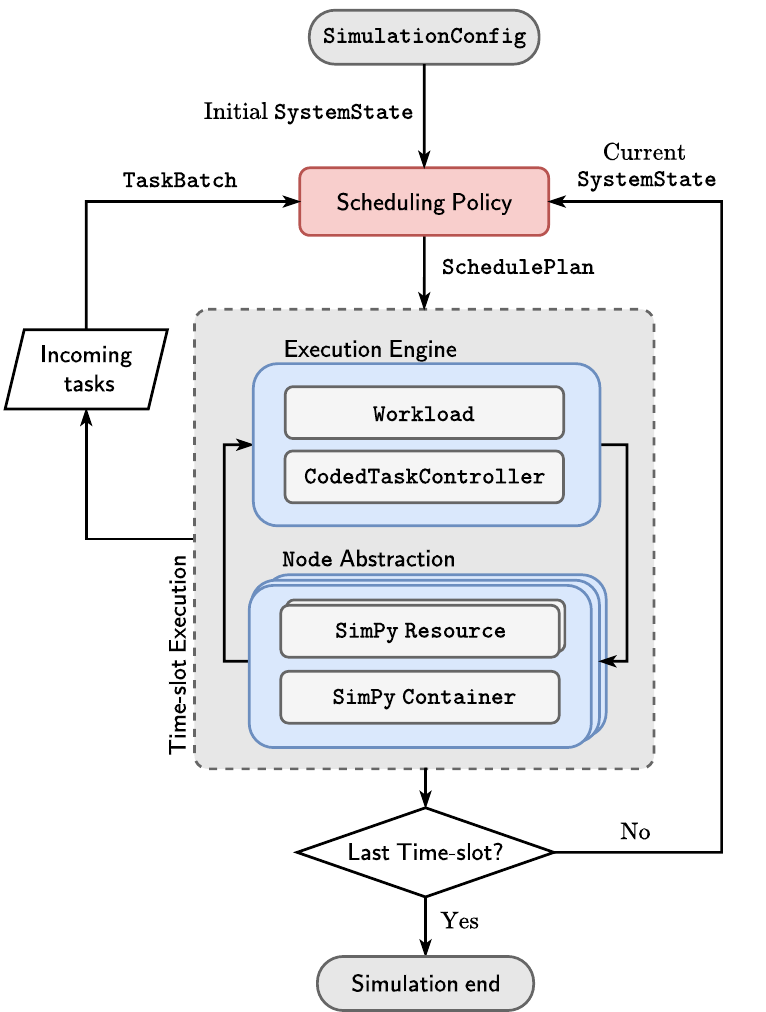}
    \caption{Workflow of proposed discrete-event simulator}
    \label{fig:DES_scheme}
\end{figure}

This simulator explicitly captures asynchronous execution, resource contention, queue concurrency, and memory-admission blocking effects. Specifically, when an offloaded share reaches a helper with insufficient available memory, the corresponding transmission is blocked until enough memory is released.

\section{Supplementary Results}

\subsection{General Performance}

\begin{figure*}[tbh]
    \centering
    \begin{subfigure}{1.0\linewidth}
        \centering
        \includegraphics[width=1.0\linewidth]{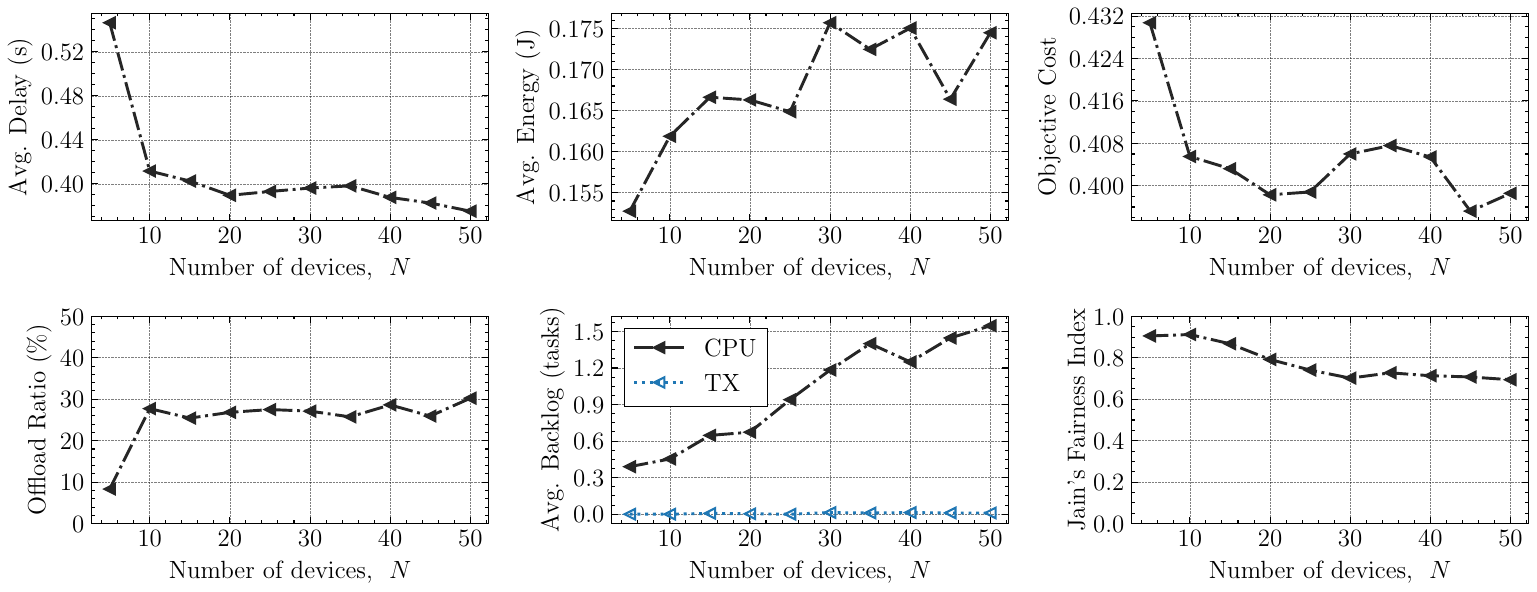}
        \caption{Sweep over number of devices}
        \label{fig:general_a}
    \end{subfigure}
    \vspace{0.3cm}
    \begin{subfigure}{1.0\linewidth}
        \centering
        \includegraphics[width=1.0\linewidth]{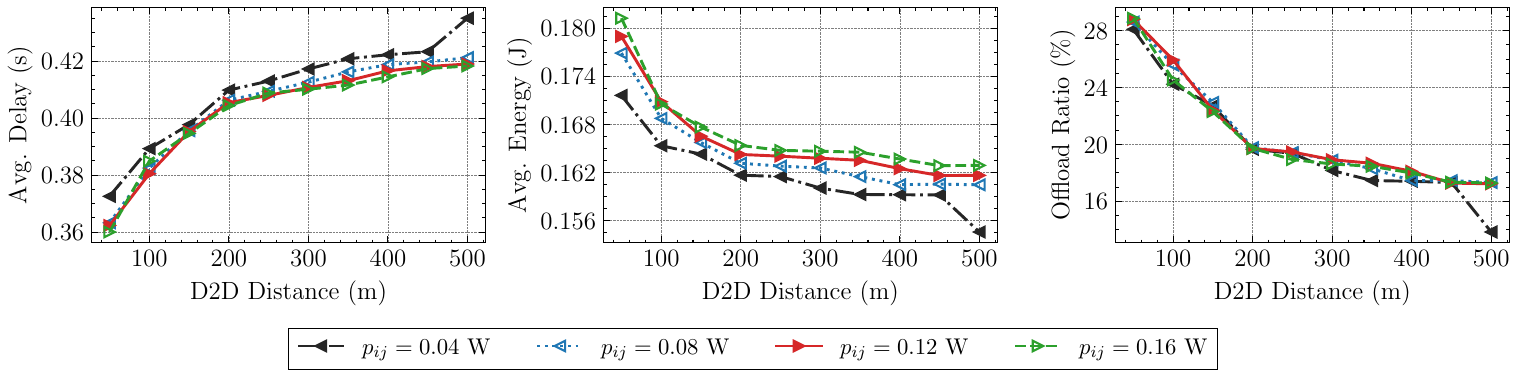}
        \caption{Sweep over D2D distance between devices}
        \label{fig:general_b}
    \end{subfigure}
    \caption{System performance under different network sizes and D2D distances. }
    \label{fig:general}
\end{figure*}

Figure~\ref{fig:general} analyzes the general system performance under different network-based parameters. Specifically, Figure~\ref{fig:general_a} evaluates the impact of the network size, measured by the number of devices. The results show a clear delay--energy trade-off as the average task delay decreases when the network size increases, going from $0.547\,\text{s}$ to $0.375\,\text{s}$, while the average task energy consumption is lower for smaller networks, going from $0.153\,\text{J}$ to $0.175\,\text{J}$. This behavior is expected since larger networks provide more potential helper devices, which improves the capacity to serve backlogged requesters and further benefits from the threshold-based recovery of the coded scheme. However, a higher number of offloaded tasks also implies a larger number of transmitted shares and remote computations. This mainly increases the backlogged computation tasks, although this effect is partially mitigated by the cancellation-aware execution policy, which avoids waiting for all shares once the recovery threshold is reached.

Moreover, even though the average number of backlogged tasks per device increases with $N$, the offload ratio remains almost constant, around $25\%$. This indicates that the system preserves a stable scheduling behavior while the absolute number of generated tasks and offloaded shares increases with the network size, which naturally produces larger queue backlogs. Additionally, the number of backlogged transmissions is almost negligible compared to the CPU backlog, which indicates that, at the end of each time slot, the transmission rate is sufficient to serve most pending transmissions. For $N=5$, the system exhibits a different behavior, with higher delay and lower offload ratio than the remaining configurations. This is caused by the limited number of balanced and helper devices under the configured device proportions, which reduces the availability of attractive offloading opportunities. 

Finally, we validate the Jain's Fairness index, which measures whether the workload is evenly distributed among devices or, conversely, computed as
\begin{equation*}
    \frac{\left( \sum_{i=1}^{N} x_i \right)^2}{N \sum_{i=1}^{N} x_i^2},
\end{equation*}
where $(x_i)_{i=1}^N$ denote the workload allocation profile across the $N$ devices, and $x_i$ represents the number of processed and waiting tasks at device $i$. It remains around $0.8$ throughout the sweep of number of devices, which shows that the heuristic solver distributes the workload fairly when selecting offloading decisions and helper devices, i.e., offloaded shares are not always assigned to the same high-capacity devices. However, when $N$ increases, this index tends to decrease. This is also expected, since the strongest devices become more attractive for offloading and progressively receive a larger fraction of the remote workload.

On the other hand, Figure~\ref{fig:general_b} analyzes the system performance when changing the D2D distance between devices. In this experiment, all devices are separated by the same distance to isolate its impact on the communication performance. The delay and offload ratio trends are consistent with the expected behavior as the average delay increases with the distance, from $0.363\,\text{s}$ to $0.421\,\text{s}$ (for $p_{ij}=0.08\,\text{W}$), since longer transmission times make offloading less attractive. Conversely, the average task energy consumption decreases with the distance because fewer tasks are offloaded, reducing the number of share transmissions and remote computations. The same trend is observed for different transmission powers, where higher powers increase the energy cost of each transmission. This can be observed in the task energy consumption, which moves from $0.154\,\text{J}$ for $p_{ij}=0.04\,\text{W}$ to $0.163\,\text{J}$ for $p_{ij}=0.16\,\text{W}$ (with $d_{ij}=500\,\text{m}$).

\subsection{System Scalability}
\label{app:system_scalability}

This subsection analyzes the system performance under different scalability configurations. Specifically, Figure~\ref{fig:scalability_arrival} evaluates the impact of different task-arrival regimes. As expected, increasing $\rho$ increases the number of backlogged tasks in the queues. The same occurs when the slot duration increases, since devices have more time to generate task requests before the next scheduling decision. The offload ratio follows a different behavior as it decreases for shorter slot durations because scheduling epochs are triggered more frequently and queues are less accumulated, which makes local execution more competitive in some slots. However, for a fixed slot duration, the offload ratio increases with higher values of $\rho$, since heavier task-arrival regimes create more opportunities to benefit from helper devices.

\begin{figure}[tbh]
    \centering
    \includegraphics[width=0.9\linewidth]{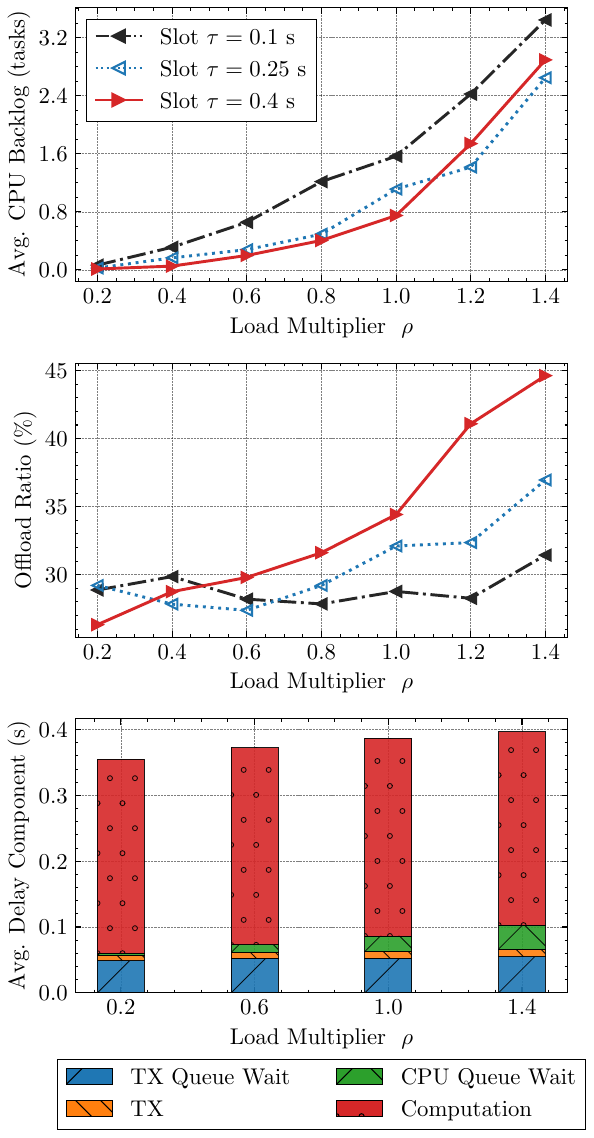}
    \caption{Scalability analysis under different task-arrival regimes and slot durations.}
    \label{fig:scalability_arrival}
\end{figure}

Additionally, we study the impact of task arrivals on the average task delay, considering its different components: transmission waiting time, transmission time, CPU waiting time, and computation time, including locally executed and offloaded workloads. The computation time and transmission waiting time remain approximately constant for different values of $\rho$. This can be explained by two factors: i) the computation time depends mainly on the workload distribution, not on the number of scheduled tasks; ii) for $\tau=0.1\,\text{s}$, the offload ratio remains mostly stable across different values of $\rho$, while the transmission service rate $R_{ij}(t)$ is fixed. In contrast, CPU waiting time clearly increases when the number of scheduled tasks grows, going from $3.453\,\text{ms}$ to $36.392\,\text{ms}$ for $\tau=0.1\,\text{s}$. The transmission time also slightly increases for higher $\rho$, since memory constraints favor the offloading of tasks with fewer shares, which makes each selected offloaded workload larger on average. Therefore, there is no direct translation between $\rho$ and task delay, i.e., reducing or increasing the task-arrival rate by a factor $\rho$ does not reduce or increase the average task delay by the same factor.

Additionally, Figure~\ref{fig:scalability_workload} analyzes the system performance under different workload distributions and slot durations. The mean workload is fixed to $1{,}000\,\text{cycles/bit}$, while the lower and upper values are modified to change the workload variability. First, the average task delay decreases when workload variability increases for two reasons: i) higher variability produces a larger number of lightweight tasks, for which local execution becomes attractive and fast; ii) the offload ratio also increases for heavier tasks, allowing the scheduler to exploit coded execution when it provides a clear delay reduction.

\begin{figure}[tbh]
    \centering
    \includegraphics[width=0.95\linewidth]{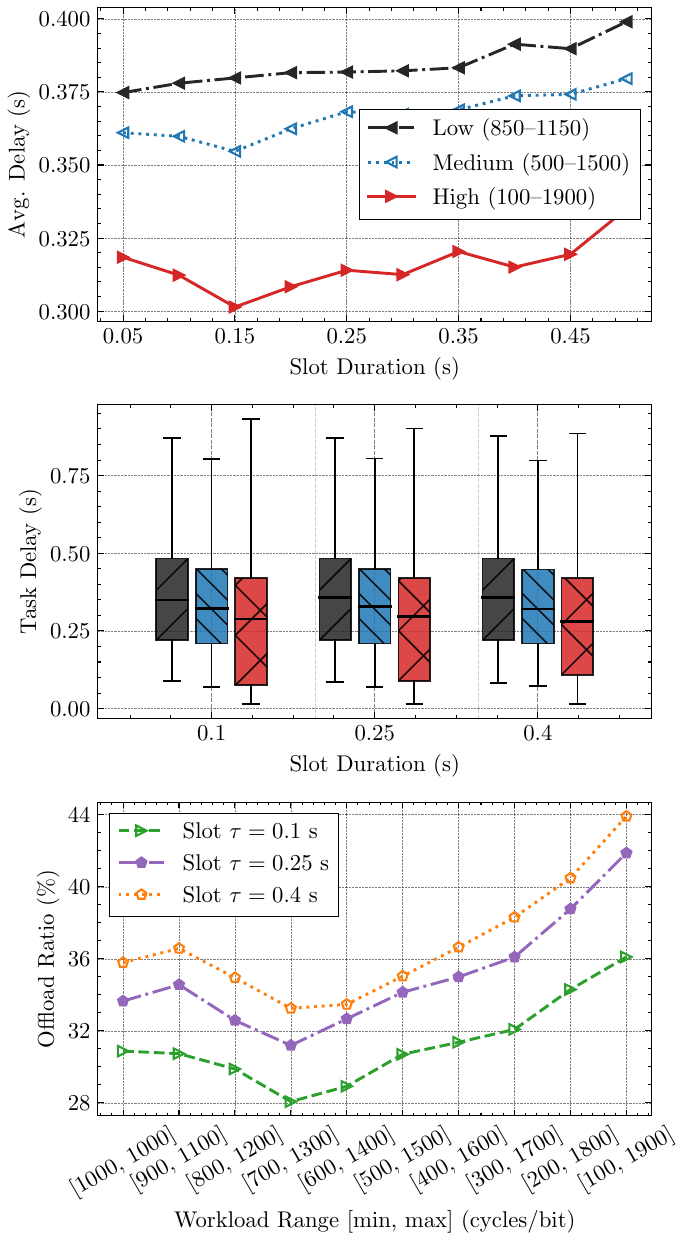}
    \caption{Scalability analysis under different workload distributions and slot durations.}
    \label{fig:scalability_workload}
\end{figure}

Moreover, as workload variability increases, the entire task-delay distribution shifts downwards. Specifically, the median delay drops from approximately $0.349\,\text{s}$ under low variability to $0.288\,\text{s}$ under high variability (for $\tau=0.1\,\text{s}$). This occurs because the wider distribution of CPU cycles per bit produces a significant fraction of lightweight tasks. These tasks are computed rapidly, as reflected by quartile-1 decreasing to $0.075$ s under high variability, which helps reduce queue pressure. Meanwhile, heavier tasks are more frequently handled through coded offloading. Comparing the boxplots across slot durations shows that the delay distributions remain stable for each workload preset. This suggests that the scheduling framework is temporally robust: as long as the system remains stable, changing the scheduling frequency does not significantly degrade the delay distribution. This stability is maintained because the system increases its offload ratio at larger slot durations, compensating for less frequent scheduling by using helper resources more aggressively.

\subsection{Delay-Energy Trade-off under Privacy Penalty}

The evaluated metrics in Figure~\ref{fig:energy_delay_privacy} are swept for different values of the pair $(\alpha,\beta)$, with $\beta=1-\alpha$. When $\alpha \to 0$, energy minimization is prioritized, which results in a negligible fraction of offloaded tasks. As a result, energy consumption is mainly limited to local computations, avoiding the coded procedure and, therefore, the overhead introduced by share transmissions and remote computations. This penalizes task delay, which reaches approximately $0.54\,\text{s}$ on average. For intermediate values of $\alpha$, e.g., $\alpha=0.5$, two main aspects can be observed: i) the impact of the $\delta$ parameter; ii) the trade-off between delay and energy. Regarding the former, lower values of $\delta$ lead to a higher offload ratio and, therefore, to higher energy consumption and lower delay. Conversely, higher values of $\delta$ reduce the attractiveness of offloading because they increase the privacy penalty. This is because $\delta \to 0$ provides almost zero privacy leakage, making offloading more favorable. Finally, when $\alpha \to 1$, delay minimization is prioritized, increasing the number of offloaded tasks to exploit coded execution and threshold-based recovery. In this case, the offload ratio remains in the range of $90\%$--$100\%$ independently of $\delta$, i.e., the privacy penalty is not enough to overcome the delay benefit of coded offloading. However, energy consumption is penalized, reaching approximately $0.255\,\text{J}$ on average.

\begin{figure}[tbh]
    \centering
    \includegraphics[width=0.95\linewidth]{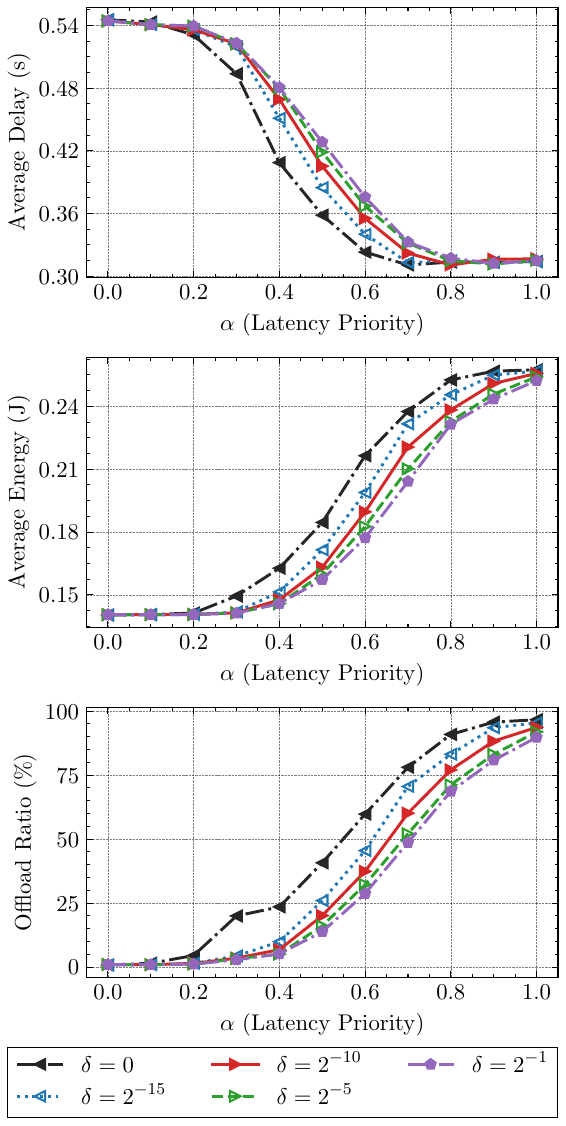}
    \caption{Delay--energy trade-off under different values of the $\delta$-noisy leakage parameter with $\gamma=0.05$.}
    \label{fig:energy_delay_privacy}
\end{figure}

\end{document}

%% file: refs.bbl

%% file: main_arXiv.bbl
\begin{thebibliography}{10}
\providecommand{\url}[1]{#1}
\csname url@samestyle\endcsname
\providecommand{\newblock}{\relax}
\providecommand{\bibinfo}[2]{#2}
\providecommand{\BIBentrySTDinterwordspacing}{\spaceskip=0pt\relax}
\providecommand{\BIBentryALTinterwordstretchfactor}{4}
\providecommand{\BIBentryALTinterwordspacing}{\spaceskip=\fontdimen2\font plus
\BIBentryALTinterwordstretchfactor\fontdimen3\font minus \fontdimen4\font\relax}
\providecommand{\BIBforeignlanguage}[2]{{%
\expandafter\ifx\csname l@#1\endcsname\relax
\typeout{** WARNING: IEEEtran.bst: No hyphenation pattern has been}%
\typeout{** loaded for the language `#1'. Using the pattern for}%
\typeout{** the default language instead.}%
\else
\language=\csname l@#1\endcsname
\fi
#2}}
\providecommand{\BIBdecl}{\relax}
\BIBdecl

\bibitem{liu2025integrated}
Z.~Liu, X.~Chen, H.~Wu, Z.~Wang, X.~Chen, D.~Niyato, and K.~Huang, ``Integrated sensing and edge ai: Realizing intelligent perception in 6g,'' \emph{IEEE Communications Surveys \& Tutorials}, 2025.

\bibitem{wen2024integrated}
D.~Wen, X.~Li, Y.~Zhou, Y.~Shi, S.~Wu, and C.~Jiang, ``Integrated sensing-communication-computation for edge artificial intelligence,'' \emph{IEEE Internet of Things Magazine}, vol.~7, no.~4, pp. 14--20, 2024.

\bibitem{gkonis2023survey}
P.~Gkonis, A.~Giannopoulos, P.~Trakadas, X.~Masip-Bruin, and F.~D’Andria, ``A survey on iot-edge-cloud continuum systems: Status, challenges, use cases, and open issues,'' \emph{Future Internet}, vol.~15, no.~12, p. 383, 2023.

\bibitem{ullah2023orchestration}
A.~Ullah, T.~Kiss, J.~Kov{\'a}cs, F.~Tusa, J.~Deslauriers, H.~Dagdeviren, R.~Arjun, and H.~Hamzeh, ``Orchestration in the cloud-to-things compute continuum: taxonomy, survey and future directions,'' \emph{Journal of Cloud Computing}, vol.~12, no.~1, pp. 1--29, 2023.

\bibitem{marchese2023application}
A.~Marchese and O.~Tomarchio, ``Application and infrastructure-aware orchestration in the cloud-to-edge continuum,'' in \emph{2023 IEEE 16th International Conference on Cloud Computing (CLOUD)}.\hskip 1em plus 0.5em minus 0.4em\relax IEEE, 2023, pp. 262--271.

\bibitem{rosmaninho2025edge}
R.~Rosmaninho, D.~Raposo, P.~Rito, and S.~Sargento, ``Edge-cloud continuum orchestration of critical services: A smart-city approach,'' \emph{IEEE Transactions on Services Computing}, 2025.

\bibitem{patel2024modeling}
Y.~S. Patel, P.~Townend, A.~Singh, and P.-O. {\"O}stberg, ``Modeling the green cloud continuum: integrating energy considerations into cloud--edge models,'' \emph{Cluster computing}, vol.~27, no.~4, pp. 4095--4125, 2024.

\bibitem{lopez2022depth}
J.~J. Lopez~Escobar, R.~P. D{\'\i}az~Redondo, and F.~Gil-Castineira, ``In-depth analysis and open challenges of mist computing,'' \emph{Journal of Cloud Computing}, vol.~11, no.~1, p.~81, 2022.

\bibitem{iorio2022computing}
M.~Iorio, F.~Risso, A.~Palesandro, L.~Camiciotti, and A.~Manzalini, ``Computing without borders: The way towards liquid computing,'' \emph{IEEE Transactions on Cloud Computing}, vol.~11, no.~3, pp. 2820--2838, 2022.

\bibitem{peng2024survey}
P.~Peng, W.~Lin, W.~Wu, H.~Zhang, S.~Peng, Q.~Wu, and K.~Li, ``A survey on computation offloading in edge systems: From the perspective of deep reinforcement learning approaches,'' \emph{Computer Science Review}, vol.~53, p. 100656, 2024.

\bibitem{pournazari2025computation}
J.~Pournazari, A.~Ullah, A.~Al-Dubai, and X.~Liu, ``Computation offloading in the edge-to-cloud compute continuum: a survey of federated architectural solutions,'' \emph{Cluster Computing}, vol.~28, no.~13, p. 839, 2025.

\bibitem{wang2025d2d}
Y.~Wang, D.~Kong, H.~Chai, H.~Qiu, R.~Xue, and S.~Li, ``D2d assisted cooperative computational offloading strategy in edge cloud computing networks,'' \emph{Scientific Reports}, vol.~15, no.~1, p. 12303, 2025.

\bibitem{wang2023location}
Z.~Wang, Y.~Sun, D.~Liu, J.~Hu, X.~Pang, Y.~Hu, and K.~Ren, ``Location privacy-aware task offloading in mobile edge computing,'' \emph{IEEE Transactions on Mobile Computing}, vol.~23, no.~3, pp. 2269--2283, 2023.

\bibitem{zhang2024novel}
D.-G. Zhang, H.-Z. An, J.~Zhang, T.~Zhang, W.-M. Dong, and X.-R. Jiang, ``Novel privacy awareness task offloading approach based on privacy entropy,'' \emph{IEEE Transactions on Network and Service Management}, vol.~21, no.~3, pp. 3598--3608, 2024.

\bibitem{li2020coded}
S.~Li and S.~Avestimehr, ``Coded computing: Mitigating fundamental bottlenecks in large-scale distributed computing and machine learning,'' \emph{Foundations and Trends in Communications and Information Theory}, vol.~17, no.~1, pp. 1--148, 2020.

\bibitem{ng2020survey}
J.~S. Ng, W.~Y.~B. Lim, N.~C. Luong, Z.~Xiong, A.~Asheralieva, D.~Niyato, C.~Leung, and C.~Miao, ``A survey of coded distributed computing,'' \emph{arXiv preprint arXiv:2008.09048}, 2020.

\bibitem{bitar2020minimizing}
R.~Bitar, P.~Parag, and S.~El~Rouayheb, ``Minimizing latency for secure coded computing using secret sharing via staircase codes,'' \emph{IEEE Transactions on Communications}, vol.~68, no.~8, pp. 4609--4619, 2020.

\bibitem{schlegel2022privacy}
R.~Schlegel, S.~Kumar, E.~Rosnes, and A.~G. i~Amat, ``Privacy-preserving coded mobile edge computing for low-latency distributed inference,'' \emph{IEEE Journal on Selected Areas in Communications}, vol.~40, no.~3, pp. 788--799, 2022.

\bibitem{beck2024survey}
S.~Beck, M.~Raavi, C.~Dale, K.~Weishalla, and B.~Worrell, ``Survey of side-channel vulnerabilities for short-range wireless communication technologies,'' in \emph{2024 IEEE International Conference on Electro Information Technology (eIT)}.\hskip 1em plus 0.5em minus 0.4em\relax IEEE, 2024, pp. 450--456.

\bibitem{gupta2026security}
U.~Gupta and H.~Mahdavifar, ``On the security of linear secret sharing with general noisy side-channel leakage,'' in \emph{Annual International Conference on the Theory and Applications of Cryptographic Techniques}.\hskip 1em plus 0.5em minus 0.4em\relax Springer, 2026, pp. 477--505.

\bibitem{sadatdiynov2023review}
K.~Sadatdiynov, L.~Cui, L.~Zhang, J.~Z. Huang, S.~Salloum, and M.~S. Mahmud, ``A review of optimization methods for computation offloading in edge computing networks,'' \emph{Digital Communications and Networks}, vol.~9, no.~2, pp. 450--461, 2023.

\bibitem{sufyan2020computation}
F.~Sufyan and A.~Banerjee, ``Computation offloading for distributed mobile edge computing network: A multiobjective approach,'' \emph{IEEE Access}, vol.~8, pp. 149\,915--149\,930, 2020.

\bibitem{jia2021energy}
M.~Jia, J.~Zhu, and H.~Huang, ``Energy and delay-ware massive task scheduling in fog-cloud computing system,'' \emph{Peer-to-Peer Networking and Applications}, vol.~14, no.~4, pp. 2139--2155, 2021.

\bibitem{niu2025pipelining}
H.~Niu, L.~Wang, K.~Du, Z.~Lu, X.~Wen, and Y.~Liu, ``A pipelining task offloading strategy via delay-aware multi-agent reinforcement learning in cybertwin-enabled 6g network,'' \emph{Digital Communications and Networks}, vol.~11, no.~1, pp. 92--105, 2025.

\bibitem{khan2023dynamic}
S.~Khan, J.~Zheng, S.~Khan, Z.~Masood, and M.~P. Akhter, ``Dynamic offloading technique for real-time edge-to-cloud computing in heterogeneous mec--mcc and iot devices,'' \emph{Internet of Things}, vol.~24, p. 100996, 2023.

\bibitem{ding2021potential}
Y.~Ding, K.~Li, C.~Liu, and K.~Li, ``A potential game theoretic approach to computation offloading strategy optimization in end-edge-cloud computing,'' \emph{IEEE Transactions on Parallel and Distributed Systems}, vol.~33, no.~6, pp. 1503--1519, 2021.

\bibitem{aggarwal2025distributed}
S.~Aggarwal, M.~Bastopcu, S.~Ulukus, T.~Ba{\c{s}}ar \emph{et~al.}, ``Distributed offloading in multi-access edge computing systems: A mean-field perspective,'' \emph{arXiv preprint arXiv:2501.18718}, 2025.

\bibitem{he2024multi}
H.~He, X.~Yang, X.~Mi, H.~Shen, and X.~Liao, ``Multi-agent deep reinforcement learning based dynamic task offloading in a device-to-device mobile-edge computing network to minimize average task delay with deadline constraints,'' \emph{Sensors}, vol.~24, no.~16, p. 5141, 2024.

\bibitem{malik2022efficient}
U.~M. Malik, M.~A. Javed, J.~Frnda, J.~Rozhon, and W.~U. Khan, ``Efficient matching-based parallel task offloading in iot networks,'' \emph{Sensors}, vol.~22, no.~18, p. 6906, 2022.

\bibitem{bitar2021private}
R.~Bitar, Y.~Xing, Y.~Keshtkarjahromi, V.~Dasari, S.~El~Rouayheb, and H.~Seferoglu, ``Private and rateless adaptive coded matrix-vector multiplication,'' \emph{EURASIP Journal on Wireless Communications and Networking}, vol. 2021, no.~1, p.~15, 2021.

\bibitem{zhang2025joint}
J.~Zhang, Y.~Hu, M.~Shao, and X.~Li, ``Joint energy-aware task offloading and privacy protection in healthcare monitoring systems via deep reinforcement learning,'' \emph{Scientific Reports}, 2025.

\bibitem{xia2024privacy}
F.~Xia, Y.~Chen, and J.~Huang, ``Privacy-preserving task offloading in mobile edge computing: A deep reinforcement learning approach,'' \emph{Software: Practice and Experience}, vol.~54, no.~9, pp. 1774--1792, 2024.

\bibitem{baek2020privacy}
H.~Baek, H.~Ko, and S.~Pack, ``Privacy-preserving and trustworthy device-to-device (d2d) offloading scheme,'' \emph{IEEE Access}, vol.~8, pp. 191\,551--191\,560, 2020.

\bibitem{li2021security}
Z.~Li, H.~Hu, H.~Hu, B.~Huang, J.~Ge, and V.~Chang, ``Security and energy-aware collaborative task offloading in d2d communication,'' \emph{Future Generation Computer Systems}, vol. 118, pp. 358--373, 2021.

\bibitem{yu2023privacy}
\BIBentryALTinterwordspacing
H.~Yu, J.~Liu, C.~Hu, and Z.~Zhu, ``Privacy-preserving task offloading strategies in {MEC},'' \emph{Sensors}, vol.~23, no.~1, p.~95, 2023. [Online]. Available: \url{https://doi.org/10.3390/s23010095}
\BIBentrySTDinterwordspacing

\bibitem{zhu2021privacy}
D.~Zhu, T.~Li, H.~Liu, J.~Sun, L.~Geng, and Y.~Liu, ``Privacy-aware online task offloading for mobile-edge computing,'' \emph{Wireless Communications and Mobile Computing}, vol. 2021, no.~1, p. 6622947, 2021.

\bibitem{razaq2021privacy}
M.~M. Razaq, B.~Tak, L.~Peng, and M.~Guizani, ``Privacy-aware collaborative task offloading in fog computing,'' \emph{IEEE Transactions on Computational Social Systems}, vol.~9, no.~1, pp. 88--96, 2021.

\bibitem{chen2007secure}
H.~Chen, R.~Cramer, S.~Goldwasser, R.~De~Haan, and V.~Vaikuntanathan, ``Secure computation from random error correcting codes,'' in \emph{Annual International Conference on the Theory and Applications of Cryptographic Techniques}.\hskip 1em plus 0.5em minus 0.4em\relax Springer, 2007, pp. 291--310.

\bibitem{zhao2017energy}
P.~Zhao, H.~Tian, C.~Qin, and G.~Nie, ``Energy-saving offloading by jointly allocating radio and computational resources for mobile edge computing,'' \emph{IEEE access}, vol.~5, pp. 11\,255--11\,268, 2017.

\bibitem{ETSI_TS_138_101_1_V16_12_1_2022}
\BIBentryALTinterwordspacing
{ETSI}, ``{5G; NR; User Equipment (UE) radio transmission and reception; Part 1: Range 1 Standalone},'' {European Telecommunications Standards Institute}, {Sophia Antipolis, France}, {Technical Specification} {ETSI TS 138 101-1 V16.12.1}, Jul. 2022, {3GPP TS 38.101-1 version 16.12.1 Release 16}. [Online]. Available: \url{https://www.etsi.org/standards-search}
\BIBentrySTDinterwordspacing

\bibitem{benhamouda2021local}
F.~Benhamouda, A.~Degwekar, Y.~Ishai, and T.~Rabin, ``On the local leakage resilience of linear secret sharing schemes: F. benhamouda et al.'' \emph{Journal of Cryptology}, vol.~34, no.~2, p.~10, 2021.

\bibitem{bruno1974scheduling}
J.~Bruno, E.~G. Coffman~Jr, and R.~Sethi, ``Scheduling independent tasks to reduce mean finishing time,'' \emph{Communications of the ACM}, vol.~17, no.~7, pp. 382--387, 1974.

\end{thebibliography}
